\documentclass[10pt]{article}

\usepackage{amssymb} 
\usepackage{amsmath, amsthm}
\usepackage{amsfonts}

\usepackage[utf8]{inputenc}
\usepackage{graphicx} 
\usepackage{color}

\usepackage{verbatim}

\usepackage[dvipsnames]{xcolor}

\usepackage{setspace}
\usepackage{fancyhdr}
\usepackage{hyperref}

\hypersetup{
  colorlinks=true,
  linkcolor=black,
  citecolor=black
}

\usepackage{authblk}

\usepackage{titlesec}
\titleformat{\section}{\normalfont\fontsize{12}{12}\bfseries}{\thesection.}{10pt}{}
\titleformat{\subsection}{\normalfont\fontsize{11}{10}\bfseries}{\thesubsection.}{5pt}{}

\usepackage[font=small,labelfont=bf]{caption}

\theoremstyle{definition}

\theoremstyle{plain}
\newtheorem{theorem}{Theorem}
\newtheorem{proposition}[theorem]{Proposition}
\newtheorem{lemma}[theorem]{Lemma}
\newtheorem{corollary}[theorem]{Corollary}

\newcommand{\arcsinh}{\operatorname{arcsinh}}

\renewcommand{\div}{\operatorname{div}}
\newcommand{\lcN}{\;\!\cal N}
\newcommand{\cN}{{\cal N}}
\newcommand\df{\,\mathrm{d}}
\newcommand{\N}{{\mathbb{N}}}
\newcommand\R{{\mathbb{R}}}
\newcommand{\MFT}{\mathrm{MFT}}
\newcommand{\lL}{\;\!L}
\renewcommand{\ss}{\mathrm{ss}}

\newcommand{\eps}{\epsilon}
\newcommand{\vhi}{\varphi}

\begin{document}

\title{Canonical structure and orthogonality of forces and currents in irreversible Markov chains}

\author[1]{Marcus Kaiser}
\author[2,3,4]{Robert L. Jack}
\author[1]{Johannes Zimmer}
\affil[1]{Department of Mathematical Sciences, University of Bath, Bath BA2 7AY, UK}
\affil[2]{Department of Applied Mathematics and Theoretical Physics, University of Cambridge, Wilberforce Road, Cambridge CB3 0WA, UK}
\affil[3]{Department of Chemistry, University of Cambridge, Lensfield Road, Cambridge CB2 1EW, UK}
\affil[4]{Department of Physics, University of Bath, Bath BA2 7AY, UK}
\maketitle

\begin{abstract}
  We discuss a canonical structure that provides a unifying description of dynamical large deviations for irreversible
  {finite state Markov chains (continuous time)}, Onsager theory, and Macroscopic Fluctuation Theory.  For Markov
  chains, this theory involves a non-linear relation between probability currents and their conjugate forces.  Within this
  framework, we show how the forces can be split into two components, which are orthogonal to each other, in a generalised sense.
  This splitting allows a decomposition of the pathwise rate function into three terms, which have physical interpretations in
  terms of dissipation and convergence to equilibrium.  Similar decompositions hold for rate functions at level 2 and level 2.5.
  These results clarify how bounds on entropy production and fluctuation theorems emerge from the underlying dynamical rules.  We
  discuss how these results for Markov chains are related to similar structures within Macroscopic Fluctuation Theory, which
  describes hydrodynamic limits of such microscopic models.
\end{abstract}

\section{Introduction}
\label{sec:intro}

We consider dynamical fluctuations in systems described by Markov chains.  The nature of such fluctuations in physical systems
constrains the mathematical models that can be used to describe them.  For example, there are well-known relationships between
equilibrium physical systems and detailed balance in Markov models~\cite[Section 5.3.4]{Gardiner2009a}.  Away from equilibrium,
fluctuation theorems~\cite{Gallavotti1995a,Jarzynski1997a,Lebowitz1999a,Maes1999a,Crooks2000a} and associated ideas of local
detailed balance~\cite{Lebowitz1999a,Maes2008b} have shown how the entropy production of a system must be accounted for correctly
when modelling physical systems.  However, the mathematical structures that determine the probabilities of non-equilibrium
fluctuations are still only partially understood.

We characterise dynamical fluctuations using an approach based on the \emph{Onsager-Machlup (OM) theory}~\cite{Machlup1953a},
which is concerned with fluctuations of macroscopic properties of physical systems (for example, density or energy).  Associated
to these fluctuations is a \emph{large-deviation principle} (LDP), which encodes the probability of rare dynamical trajectories.
The classical ideas of OM theory have been extended in recent years, through the \emph{Macroscopic Fluctuation Theory} (MFT) of
Bertini \emph{et al.}~\cite{Bertini2015a}.  This theory uses an LDP to describe path probabilities for the density and current in
diffusive systems, on the hydrodynamic scale.  At the centre of MFT is a decomposition of the current into two orthogonal terms,
one of which is symmetric under time-reversal, and another which is anti-symmetric. The resulting theory is a general framework
for the analysis of dynamical fluctuations in a large class of non-equilibrium systems.  It also connects dynamical fluctuations
with thermodynamic quantities like free energy and entropy production, and with associated non-equilibrium objects like the
quasi-potential (which extends the thermodynamic free energy to non-equilibrium settings).

Here, we show how several features that appear in MFT can be attributed to a general structure that characterises dynamical
fluctuations in microscopic Markov models.  That is, the properties of the hydrodynamic (MFT) theory can be traced back to the
properties of the underlying stochastic processes.  Our approach builds on recent work by Mielke, Renger and M.~A.~Peletier, in
which the analogue of the OM theory for reversible Markov chains has been described in terms of a \emph{generalised gradient-flow
  structure}~\cite{Mielke2014a}.  To describe non-equilibrium processes, that theory must be generalised to include irreversible
Markov chains. This can be achieved using the canonical structure of fluctuations discovered by Maes and
Neto{\v{c}}n\'y~\cite{Maes2008a}.  Extending their approach, we decompose currents in the system into two parts, and we identify a
kind of orthogonality relationship associated with this decomposition.  However, in contrast to the classical OM theory and to
MFT, the large deviation principles that appear in our approach have non-quadratic rate functions, which means that fluxes have
non-linear dependence on their conjugate forces. Thus, the idea of orthogonality between currents needs to be generalised, just as
the notion of gradient flows in macroscopic equilibrium systems can be extended to generalised gradient flows.

The central players in our analysis are the probability density $\rho$ and the probability current $j$. For a given Markov chain,
the relation between these quantities is fully encoded in the master equation, which also fully specifies the dynamical
fluctuations in that model.  However, thermodynamic aspects of the system --- the roles of heat, free energy, and entropy
production --- are not apparent in the master equation.  Within the Onsager-Machlup theory, these thermodynamic quantities appear
in the action functional for paths, and solutions of the master equation appear as paths of minimal action.  Hence, the structure
that we discuss here, and particularly the decomposition of the current into two components, links the dynamical properties of the
system to thermodynamic concepts, both for equilibrium and non-equilibrium systems.

\subsection{Summary}
\label{sec:Summary}

We now sketch the setting considered in this article (precise definitions of the systems of interest and the relevant currents,
densities and forces will be given in Section~\ref{sec:Onsag-Machl-theory-Markov} below).

We introduce a large parameter $\cal N$, which might be the {size of} the system (as in MFT) or a large number of
copies of the system (an ensemble), as considered for Markov chains in~\cite{Maes2008b}. Then let
$(\hat\rho^{\lcN}_t,\hat\jmath^{\lcN}_t)_{t\in[0,T]}$ be the (random) path followed by the system's density and current, in the
time interval $[0,T]$.  Consider a random initial condition such that
$\mathrm{Prob}( \hat\rho_0^{\lcN} \approx \rho ) \asymp \exp[-\mathcal{N} I_0(\rho) ]$, asymptotically as $\mathcal N\to\infty$,
for some rate functional $I_0$. Paths that in addition satisfy a continuity equation $\dot\rho + \operatorname{div} j =0$ have the
asymptotic probability
\begin{equation}
  \label{equ:pathwise-general}
  \mathrm{Prob}\left( (\hat\rho_t^{\lcN},\hat\jmath_t^{\lcN})_{t\in[0,T]} \approx (\rho_t, j_t)_{t\in[0,T]}\right)
  \asymp  \exp\left\{-\mathcal N I_{[0,T]}\left( (\rho_t, j_t)_{t\in[0,T]}\right)\right\}
\end{equation}
with the \emph{rate functional}
\begin{equation}
  \label{eqn:mc_rate_functional}
  I_{[0,T]}\bigl((\rho_t, j_t)_{t\in[0,T]}\bigr)=  I_0(\rho_0) + \frac12\int_0^T \Phi(\rho_t,j_t, F(\rho_t))  \df t ;
\end{equation}
{here $F(\rho_t)$ is a force (see~\eqref{equ:def-aF} below for the precise definition)} and $\Phi$ is what we call the \emph{generalised OM functional},
{which has the general form}
\begin{equation}
  \label{eqn:Phi_function}
  \Phi(\rho,j,f):= \Psi(\rho,j) - j \cdot f + \Psi^\star(\rho, f),
\end{equation} 
where $j\cdot f$ is a dual pairing between {a current $j$ and a force $f$}, while $\Psi$ and $\Psi^\star$ are a pair of functions which satisfy
\begin{equation}
  \label{equ:legend}
  \Psi^\star(\rho,f) = \sup_j\bigl[  j\cdot f - \Psi(\rho,j)\bigr],\quad\text{and}\quad
  \Psi(\rho,j) = \sup_f\bigl[  j\cdot f - \Psi^\star(\rho,f)\bigr],
\end{equation}
{as well as $\Psi^\star(\rho,f)=\Psi^\star(\rho,-f)$ and $\Psi(\rho,j)=\Psi(\rho,-j)$.
  Note that~\eqref{equ:legend} means that the two functions satisfy a Legendre duality.
  Moreover, these two functions $\Psi$ and $\Psi^\star$ are strictly convex in their second arguments.  Here and throughout, $f$ indicates a force, while $F$ is a function whose (density-dependent) value is a force.

  The large deviation principle stated in~\eqref{equ:pathwise-general} is somewhat abstract: for example, $\hat\rho_t^{\lcN}$ might be defined as a density on a discrete space or on $\mathbb{R}^d$, depending on the system of interest.  Specific examples will be given below.  In addition, all microscopic parameters of the system (particle hopping rates, diffusion constants, etc.) will enter the (system-dependent) functions $\Psi$, $\Psi^\star$ and $F$.

As a preliminary example, we recall} the classical Onsager theory~\cite{Machlup1953a}, in which one considers $n$ currents
$j=(j^\alpha)_{\alpha=1}^n$ and a set of conjugate {applied forces $F=(F^\alpha)_{\alpha=1}^n$. Examples of currents
  might be particle flow or heat flow, and the relevant forces might be pressure or temperature gradients.  The large parameter
  $\cN$ corresponds to the size of a macroscopic system. The theory aims to to describe the typical (average) response of the
  current $j$ to the force $F$, and also the fluctuations of $j$.}  In this (simplest) case, the density $\rho$ plays no role,
{so the force $F$ has a fixed value in $\mathbb{R}^n$}.  The dual pairing is simply
$j\cdot f = \sum_\alpha j^\alpha f^\alpha$ and $\Psi$ is given by
$\Psi(\rho,j)=\frac12 \sum_{\alpha,\beta} j^\alpha R^{\alpha\beta} j^\beta$, where $R$ is a symmetric $n\times n$ matrix with
elements $R^{\alpha\beta}$.  The Legendre dual of $\Psi$ is
$\Psi^\star(\rho,f) =\frac12 \sum_{\alpha,\beta} f^\alpha L^{\alpha\beta} f^\beta$, where $L=R^{-1}$ is the \emph{Onsager matrix},
{whose elements are the linear response coefficients of the system}.  One sees that $\Psi$ and $\Psi^\star$ can be
interpreted as squared norms for currents and forces respectively.  Denoting this norm by $\| j \|^2_{L^{-1}} := \Psi(\rho,j)$,
one has
\begin{equation}
  \Phi(\rho,j,f) = \| j - L f \|^2_{L^{-1}}.
\end{equation}
{On applying an external force $F$, the response of the current $j$} is obtained as the minimum of $\Phi$, so $j=LF$ (that is, $j^\alpha = \sum_\beta
L^{\alpha\beta} F^\beta$).  One sees that $\Phi$ measures the deviation of the current $j$ from its expected value $LF$, within an
appropriate norm. {From the LDP~\eqref{equ:pathwise-general}, one sees that the size of this deviation determines the probability of observing a current fluctuation of this size.}

In this article, we show in Section~\ref{sec:Onsag-Machl-theory-Markov} that {finite} Markov chains have an LDP rate
functional of the form~\eqref{eqn:Phi_function}, where $\Phi$ (and thus $\Psi^\star$) are \emph{not} quadratic.  {In
  that case, $\rho$ and $j$ correspond to probability densities and probability currents, while the transition rates of the Markov
  chain determine the functions $F$, $\Psi$ and $\Psi^\star$.}  Since $\Psi$ and $\Psi^\star$ measure respectively the sizes of
the currents and forces, we interpret them as generalisations of the squared norms that appear in the classical case.  The
resulting $\Phi$ is not a squared norm, but it is still a non-negative function that measures the deviation of $j$ from its most
likely value. This leads to nonlinear relations between forces and currents.  The MFT theory~\cite{Bertini2015a} also fits in this
framework, as we show in Section~\ref{sec:Connections-to-MFT}: {in that case $\rho,j$ are a particle density and a
  particle current.  However, there are relationships between the functions $\Phi$ for MFT and for general Markov chains, as we
  discuss in Section~\ref{sec:hydro}.}

Hence, the general structure of Equs.~\eqref{equ:pathwise-general}-\eqref{equ:legend} describes classical OM
theory~\cite{Machlup1953a}, MFT, and {finite} Markov chains. A benefit is that the terms have a physical
interpretation. For a path $(\rho, j)$, the time-reversed path is $(\rho^*_t,j^*_t):=(\rho_{T-t},-j_{T-t})$. Since both $\Psi$ and
$\Psi^\star$ are symmetric in their second argument and thus invariant under time reversal, it holds that
{$\Phi(\rho,j,f) - \Phi(\rho^*,j^*,f)  = - 2 j \cdot f$.} This allows us to identify
{$j\cdot F(\rho)$} as a rate of entropy production. In contrast, the term
$\Psi(\rho,j) + \Psi^\star(\rho, {F(\rho)})$ is symmetric under time reversal and encodes the frenesy
(see~\cite{Basu2015a}). Thus, within this general structure, the physical significance of
Equations~\eqref{equ:pathwise-general}--\eqref{equ:legend} is that they connect path probabilities to physical notions such as
force, current, entropy production and breaking of time-reversal symmetry. Furthermore, we introduce in
Section~\ref{sec:Decomp-forc-rate} decompositions of forces and the (path-wise) rate
functional. Section~\ref{sec:Connections-to-MFT} shows that some results of MFT originate from generalised orthogonalities of the
underlying Markov chains derived in Section~\ref{sec:Decomp-forc-rate}. Similar results hold for time-average large deviation
principles, as shown in Section~\ref{sec:LDPs-time-averaged}. In Section~\ref{sec:Cons-struct-OM}, we show how some properties of
MFT can be derived directly from the canonical structure~\eqref{equ:pathwise-general}--\eqref{equ:legend}, independent of the
specific models of interest.  Hence these results of MFT have analogues in Markov chains.  Finally we briefly summarise our
conclusions in Section~\ref{sec:conc}.

\section{Onsager-Machlup theory for Markov chains}
\label{sec:Onsag-Machl-theory-Markov}

In this section, we collect results on forces and currents in Markov chains and on associated LDPs. In particular, we recall the
setting of~\cite{Maes2008b,Maes2008a}; other references for this section are for example~\cite{Schnakenberg1976a} (for the
definition of forces and currents in Markov chains) and~\cite{Mielke2014a} for LDPs.

\subsection{Setting}
\label{sec:Setting}

We consider an irreducible continuous time Markov chain {$X_t$} on a {finite} state space $V$ with a unique stationary
distribution $\pi$ that satisfies $\pi(x)>0$ for all $x\in V$. The transition rate from state $x$ to state $y$ is denoted with
$r_{xy}$. We assume that $r_{xy}>0$ if and only if $r_{yx}>0$.

{We restrict to finite Markov chains for simplicity: the theory can be extended to countable state Markov chains, but
  this requires some additional assumptions.  Briefly, one requires that the Markov chain should be positively recurrent and
  ergodic (see for instance~\cite{Bertini2015b}), for which it is sufficient that (i) the transition rates are not degenerate:
  $\sum_{y\in V}r_{xy}<\infty$ for all $x\in V$, and (ii) for each $x\in V$, the Markov chain started in $x$ almost all
  trajectories of the Markov chain do not exhibit infinitely many jumps in finite time (``no explosion'').  Second, one has to
  invoke a summability condition for the currents considered below (see, e.g., equations~\eqref{equ:master} and~\eqref{eq:div}),
  such that in particular the discrete integration by parts (or summation by parts) formula~\eqref{equ:parts} holds.  Finally,
  note that the cited result for existence and uniqueness of the optimal control potential (the solution
  to~\eqref{equ:vhi-balance}) is only valid for finite state Markov chains.}

As usual, we can interpret the state space of the Markov chain as a directed graph with vertices $V$ and edges
$E=\left\{xy \bigm| x,y\in V, r_{xy}>0\right\}$, such that $xy\in E$ if and only if $yx\in E$. Let $\rho$ be a probability measure
on $V$.  We define rescaled transition rates with respect to $\pi$ as
\begin{equation}
  q_{xy}:=\pi(x)r_{xy},
\end{equation}
so that $\rho(x)r_{xy} = \tfrac{\rho(x)}{\pi(x)}q_{xy}$. With this notation, the \emph{detailed balance} condition
$\pi(x) r_{xy} = \pi(y) r_{yx}$ reads $q_{xy} = q_{yx}$, so this equality holds precisely if the Markov chain is reversible
(i.e.~satisfies detailed balance). In general (not assuming reversibility), since $\pi$ is the invariant measure for the Markov
chain, one has (for all $x$) that
\begin{equation}
  \label{equ:q-balance}
  \sum_y (q_{xy} - q_{yx} ) = 0 .
\end{equation}

We further define the \emph{free energy} $\mathcal F$ on $V$ to be the \emph{relative entropy} (or \emph{Kullback-Leibler
  divergence}) with respect to $\pi$,
\begin{equation}
  \label{eqn:free_energy}
  \mathcal F(\rho) := \sum_{x} \rho(x) \log \Bigl(\frac{\rho(x)}{\pi(x)}\Bigr).
\end{equation}
The \emph{probability current} $J(\rho)$ is defined as~\cite[Equation~(7.4)]{Schnakenberg1976a}
\begin{equation}
  \label{equ:master}
  J_{xy}(\rho) := \rho(x)r_{xy}-\rho(y)r_{yx}  .
\end{equation}

Moreover, for a general current $j$ such that $j_{xy}=-j_{yx}$, we define the \emph{divergence} as
\begin{equation}
  \label{eq:div}
  \div j(x) := \sum_{y\in V} j_{xy}. 
\end{equation}
We say that $j$ is \emph{divergence free} if $\div j(x) = 0$ for every $x \in V$. The time evolution of the probability density
$\rho$ is then given by the master equation
\begin{equation}
  \label{eq:master}
  \dot \rho_t = -\div J(\rho_t)
\end{equation}
(which is often stated as $\dot\rho_t = \mathcal L^\dag \rho_t$, with the (forward) generator $\mathcal L^\dag$).

\subsection{Non-linear flux-force relation and the associated functionals \texorpdfstring{$\Psi$ and $\Psi^\star$}{}}
\label{sec:Non-linear-flux}

To apply the theory outlined in Section~\ref{sec:Summary}, the next step is to identify the appropriate forces
{$F(\rho)$ and also a set of mobilities $a(\rho)$}. In this section we define these forces,
following~\cite{Schnakenberg1976a,Maes2008b,Maes2008a}. {This amounts to a reparameterisation of the rates of the
  Markov process in terms of physically-relevant variables: an example is given in Section~\ref{sec:ring}.}

To each edge in $E$ we assign a \emph{force} $F$ and a \emph{mobility} $a$, as
\begin{equation}
  \label{equ:def-aF}
  F_{xy}(\rho) := \log \frac{\rho(x) r_{xy} }{ \rho(y) r_{yx} } \quad\text{and}\quad
  a_{xy}(\rho) := 2\sqrt{ \rho(x) r_{xy} \rho(y) r_{yx} }.
\end{equation}
Note that $F_{xy}=-F_{yx}$, while $a_{xy}=a_{yx}$: forces have a direction but the mobility is a symmetric property of each edge.
The fact that $F_{xy}$ depends on the density $\rho$ means that these forces act in the space of probability distributions.  This
definition of the force is sometimes also called \emph{affinity}~\cite[Equation~(7.5)]{Schnakenberg1976a}; see
also~\cite{Andrieux2007a}.  With this definition, the probability current~\eqref{equ:master} is
\begin{equation}
  \label{eq:J-sinh}
  J_{xy}(\rho) = a_{xy}(\rho) \sinh \bigl( \tfrac12 F_{xy}(\rho) \bigr) ,
\end{equation}
which may be verified directly from the definition $\sinh(x) = ({\rm e}^x-{\rm e}^{-x})/2$.  In contrast to the classical OM
theory, this is a \emph{non-linear} relation between forces and fluxes, although one recovers a linear structure for small forces
(recall the classical theory in Section~\ref{sec:Summary}, for which $j=Lf$).

Now consider a current $j$ defined on $E$, with {$j_{xy}=-j_{yx}$}, and a general force $f$ that satisfies
$f_{xy}=-f_{yx}$ (which is not in general given by~\eqref{equ:def-aF}).  Define a dual pair on $E$ as
\begin{equation}
  \label{eqn:dual_pairing}
  j\cdot f := \frac 12\sum_{xy} j_{xy}f_{xy},
\end{equation} 
where the summation is over all $xy\in E$ (the normalisation $1/2$ appears because each connected pair of states should be counted
only once, but $E$ is a set of directed edges, so it contains both $xy$ and $yx$, which have the same contribution to $j\cdot f$).

We define the discrete gradient {$\nabla g$ by $\nabla^{x,y}g:=g(y)-g(x)$. The discrete gradient and the divergence
  defined in~\eqref{eq:div} satisfy a discrete integration by parts formula: for any function $g\colon V\to\mathbb{R}$, since
  $j_{xy} = -j_{yx}$, we have
\begin{equation}
  \label{equ:parts}
  -\sum_{x \in V} g(x) \div j(x) = \frac 12 \sum_{xy} j_{xy} \nabla^{x,y} g = j\cdot \nabla g.
\end{equation}
} We will show in Section~\ref{sec:Large-Devi-Onsag} that there is an OM functional associated with these forces and currents,
which is of the form~\eqref{eqn:Phi_function}. Since $\Psi$ and $\Psi^\star$ are convex and related by a Legendre transformation,
it is sufficient to specify only one of them. The appropriate choice turns out to be
\begin{equation}
  \label{eqn:psi_star}
  \Psi^\star(\rho,f):=\sum_{xy}a_{xy}(\rho) \bigl(\cosh\bigl( \tfrac12 f_{xy} \bigr)-1\bigr).
\end{equation}
This means that $\Phi(\rho,j,f)$ defined in~\eqref{eqn:Phi_function} is uniquely minimised for the current
$j_{xy} = j^f_{xy}(\rho)$ with
\begin{equation}
{ j^f_{xy}(\rho) = 2(\delta\Psi^\star/\delta f)_{xy} = a_{xy}(\rho) \sinh(f_{xy}/2), }
\label{equ:def-jF}
\end{equation} 
as required for consistency with~\eqref{eq:J-sinh}. From~\eqref{equ:legend} and~\eqref{eqn:dual_pairing}, one has also
\begin{equation}
  \label{eqn:new_psi}
  \Psi(\rho,j)= \frac 12 \sum_{xy} j_{xy}f^j_{xy}(\rho) - 
  \sum_{xy} a_{xy}(\rho)\bigl(\cosh\bigl(\tfrac12 f^j_{xy}(\rho) \bigr)-1\bigr),
\end{equation}
where 
\begin{equation}
{f^j_{xy}(\rho):=2\arcsinh\left({j_{xy}/a_{xy}(\rho)}\right) }
\end{equation}
is the force required to induce the current $j$.

Physically, $\Psi^\star(\rho,f)$ is a measure of the strength of the force $f$ and $\Psi(\rho,j)$ is a measure of the magnitude of
the current $j$.  Consistent with this interpretation, note that $\Psi$ and $\Psi^\star$ are symmetric in their second arguments.
Moreover, for small forces and currents, $\Psi^\star$ and $\Psi$ are quadratic in their second arguments, and can be interpreted
as generalisations of squared norms of the force and current respectively. {Note that equations~\eqref{eqn:psi_star}
  and~\eqref{eqn:new_psi} can alternatively be represented as
\begin{equation}
\Psi(\rho,j)= \sum_{xy} \biggl[ \frac 12j_{xy}f^j_{xy}(\rho) - \sqrt{j_{xy}^2+a_{xy}(\rho)^2} + a_{xy}(\rho)\biggr]
\end{equation}
and
\begin{equation}
\Psi^\star(\rho,f):=\sum_{xy}\biggl[ \sqrt{j^f_{xy}(\rho)^2+a_{xy}(\rho)^2} - a_{xy}(\rho)\biggr].
\end{equation}
}

\subsection{Large Deviations and the Onsager-Machlup functional}
\label{sec:Large-Devi-Onsag}

As anticipated in Section~\ref{sec:Summary}, the motivation for the definitions of $\Psi$, $\Psi^\star$, and $F$ is that there is a
large deviation principle for these Markov chains, whose rate function is of the form given in~\eqref{eqn:mc_rate_functional}.
This large deviation principle appears when one considers $\cal N$ {independent} copies of the Markov chain.

We denote the $i$-th copy of the Markov chain by $X^i_t$ and define the empirical density for this copy as
$\hat\rho^{\;\!i}_t(x)=\delta_{X^i_t,x}$, where $\delta$ is a Kronecker delta function.  Let the times at which the Markov chain
$X^i_t$ has jumps in $[0,T]$ be $t_1^i, t_2^i, \dots, t^i_{K_i}$.  Further denote the state just before the $k$-th jump with
$x_{k-1}^i$ (such that the state after the $k$-th jump is $x_{k}^i$).  With this, the empirical current is given by
\begin{equation*}
  (\hat \jmath_t^{\;\!i})_{xy} 
  = \sum_{k=1}^{K_i} \bigl(\delta_{x,x_{k-1}^i} \delta_{y,x_k^i} - \delta_{y,x_{k-1}^i} \delta_{x,x_{k}^i}\bigr) \delta\bigl(t-t_k^i\bigr) , 
\end{equation*}
where $\delta(t-t_k)$ denotes a Dirac delta. Note that $(\hat\jmath_t^{\;\! i})_{xy}=-(\hat\jmath_t^{\;\! i})_{yx}$ and the total
probability is conserved, {as} $\sum_x \div \hat\jmath_t^{\;\! i}(x) =0$ {(which holds for any discrete
  vector field with $(\hat\jmath_t^{\;\! i})_{xy}=-(\hat\jmath_t^{\;\! i})_{yx}$)}. With a slight abuse of notation we define a
similar empirical {density and current} for the full set of copies as
\begin{equation}
  \label{eqn:N_average}
  \hat\rho_t^{\;\!\lcN}:= \frac 1{\cN}\sum_{i=1}^\cN \hat\rho^{\;\!i}_t, \quad\text{and}\quad
  \hat\jmath_t^{\;\!\lcN} := \frac 1{\cN}\sum_{i=1}^\cN 
  \hat\jmath^{\;\!i}_t.
\end{equation}

Next, we state the large deviation principle where the OM functional appears. For this, we fix a time interval $[0,T]$ and
consider the large $\mathcal N$ limit.  We assume that the $\cal N$ copies at time $t=0$ have initial conditions drawn from the
invariant measure of the process (the generalisation to other initial conditions is straightforward). Then, the probability to
observe a joint density and current $(\rho_t, j_t)_{t\in[0,T]}$ over the time interval $[0,T]$ is in the limit as
$\mathcal N\to\infty$ given by~\eqref{equ:pathwise-general}. That is,
\begin{equation}
  \label{eqn:ldp_pathwise_statement}
  \mathrm{Prob}\Bigl( (\hat\rho_t^{\;\!\mathcal N},\hat\jmath_t^{\;\!\mathcal N})_{t\in[0,T]} \approx (\rho_t, j_t)_{t\in[0,T]}\Bigr)
  \asymp  \exp\bigl\{-\mathcal N I_{[0,T]}\bigl( (\rho_t, j_t)_{t\in[0,T]}\bigr)\bigr\}
\end{equation}
with
\begin{equation}
  \label{eqn:ldp_pathwise}
  I_{[0,T]}\bigl((\rho_t,j_t)_{t\in[0,T]}\bigr) =
  \begin{cases} 
    \mathcal{F}(\rho_0) + \frac 12\int_0^T \Phi(\rho_t,j_t, F(\rho_t)) \df t    
    & \text{if } \dot\rho_t + \operatorname{div} j_t = 0\\ 
    +\infty & \text{otherwise}
\end{cases} 
\end{equation}
Here, $F(\rho)$ is the force defined in~\eqref{equ:def-aF} and the condition $\dot\rho_t + \div j_t = 0$ has to hold for almost
all $t\in[0,T]$. Moreover, $\Phi$ is of the form $\Phi(\rho,j,f) = \Psi(\rho,j) - j\cdot f + \Psi^\star(\rho,f)$ stated
in~\eqref{eqn:Phi_function}, and the relevant functions $\Psi$, $\Psi^\star$ and $\cal F$ are those of~\eqref{eqn:psi_star},
\eqref{eqn:new_psi} and~\eqref{eqn:free_energy}.  This LDP was formally derived in~\cite{Maes2008a,Maes2008b}. Since the
quantities defined in~\eqref{eqn:N_average} are simple averages over independent copies of the same Markov chain, this LDP may
also be proven by direct application of Sanov's theorem, which provides an interpretation of $I_{[0,T]}$ as a relative entropy
between path measures; we sketch the derivation in Appendix~\ref{sec:relent}. For finite-state Markov chains,
\eqref{eqn:ldp_pathwise_statement} and~\eqref{eqn:ldp_pathwise} also follow (by contraction) from~\cite[Theorem 4.2]{Renger2017a},
which provides a rigorous proof.

{We emphasise that the arguments $\rho$ and $j$ of the function $\Phi$ correspond to the random variables that appear
  in the LDP, while the functions $F$, $\Psi$ and $\Psi^\star$ that appear in $\Phi$ encapsulate the transition rates of the
  Markov chain.  Thus, by reparameterising the rates $r_{xy}$ in terms of forces $F$ and mobilities $a$, we arrive at a
  representation of the rate function which helps to make its properties transparent (convexity, positivity, symmetries such
  as~\eqref{equ:gc-finite-time}).}

We note that for reversible Markov chains, the force $F(\rho)$ is a pure gradient $F=\nabla G$ for some potential $G$ (see
Section~\ref{sec:Decomp-forc-rate} below), in which case one may write $j\cdot F=\sum_x \dot\rho(x) G(x)$, which follows from an
integration by parts and application of the continuity equation.  In this case, Mielke, M.~A.~Peletier, and
Renger~\cite{Mielke2014a} also identified a slightly different canonical structure to the one presented here, in which the dual
pairing is $\sum_x v(x) G(x)$, for a velocity $v(x)=\dot\rho(x)$ and a potential $G$.  The analogues of $\Psi$ and $\Psi^\star$ in
that setting depend on $v$ and $G$ respectively, instead of $j$ and $F$.  The setting of~\eqref{eqn:Phi_function}
and~\eqref{equ:legend} is more general, in that the functions $\Psi,\Psi^\star$ for the velocity/potential setting are fully
determined by those for the current/force setting.  Also, focusing on the velocity $v$ prevents any analysis of the
divergence-free part of the current, and restricting to potential forces does not generalise in a simple way to irreversible
Markov chains.  For this reason, we use the current/force setting in this work.

In a separate development, Maas~\cite{Maas2011a} identified a quadratic cost function for paths (in fact a metric structure) for
which the master equation~\eqref{eq:master} is the minimiser in the case of reversible dynamics. This metric corresponds to the
solution of an optimal mass transfer problem which seems to have no straightforward extension to irreversible systems.  Of course,
in the reversible case, the pathwise rate function~\eqref{eqn:ldp_pathwise} has the same minimiser, but is non-quadratic and
therefore does not correspond to a metric structure, so there is no simple geometrical interpretation of~\eqref{eqn:ldp_pathwise}.
It seems that the non-quadratic structure in the rate function is essential in order capture the large deviations encoded
by~\eqref{eqn:ldp_pathwise_statement}.

\subsection{Time-reversal symmetry, entropy production, and the Gallavotti-Cohen theorem}
\label{sec:Time-reversal-symm}

The rate function for the large-deviation principle~\eqref{eqn:ldp_pathwise_statement} is given by~\eqref{eqn:ldp_pathwise}, which
has been written in terms of forces $F$, currents $j$, and densities $\rho$.  To explain why it is useful to write the rate
function in this way, we compare the probability of a path $(\rho_t,j_t)_{t\in[0,T]}$ with that of its time-reversed counterpart
$(\rho_t^*,j_t^*)_{t\in[0,T]}$, where $(\rho^*_t,j^*_t) =(\rho_{T-t},-j_{T-t})$ as before.

In this case, the fact that $\Psi$ and $\Psi^\star$ are both even in their second argument means that {\begin{align}
    &\hspace{-30pt}-\frac{1}{\cN} \log\frac{ \mathrm{Prob}\Bigl( (\hat\rho_t^{\;\!\mathcal N},\hat\jmath_t^{\;\!\mathcal
        N})_{t\in[0,T]} \approx (\rho_t, j_t)_{t\in[0,T]}\Bigr) } {\mathrm{Prob}\Bigl( (\hat\rho_t^{\;\!\mathcal
        N},\hat\jmath_t^{\;\!\mathcal N})_{t\in[0,T]} \approx (\rho_t^*, j_t^*)_{t\in[0,T]}\Bigr)
      } \nonumber \\
    &\hspace{10pt} \asymp I_{[0,T]}\bigl( (\rho_t, j_t)_{t\in[0,T]}\bigr) - I_{[0,T]}\bigl( (\rho_t^*, j_t^*)_{t\in[0,T]}\bigr)
      \nonumber \\
    & \hspace{10pt} = \mathcal F(\rho_0) - \mathcal F(\rho_T) - \int_0^T j_t\cdot F(\rho_t)\df t .
    \label{equ:gc-finite-time}
\end{align}
} This formula is a (finite-time) statement of the Gallavotti-Cohen fluctuation theorem~\cite{Gallavotti1995a,Lebowitz1999a}: see
also~\cite{Crooks2000a,Maes1999a}. It also provides a connection to physical properties of the system being modelled, via the
theory of stochastic thermodynamics~\cite{Seifert2012a}.  The terms involving the free energy $\cal F$ come from the initial
conditions of the forward and reverse paths, while the integral of $j\cdot F$ corresponds to the heat transferred from the system
to its environment during the trajectory~\cite[Eqs.~(18), (20)]{Seifert2012a}. This latter quantity -- which is the time-reversal
antisymmetric part of the pathwise rate function -- is related (by a factor of the environmental temperature) to the entropy
production in the environment~\cite{Maes1999a}.  The definition of the force $F$ in~\eqref{equ:def-aF} has been chosen so that the
dual pairing $j\cdot F$ is equal to this rate of heat flow: this means that the forces and currents are conjugate variables, just
as (for example) pressure and volume are conjugate in equilibrium thermodynamics.  {See also the example in
  Section~\ref{sec:ring} below.}

\section{Decomposition of forces and rate functional}
\label{sec:Decomp-forc-rate}

We now introduce a splitting of the force $F(\rho)$ into two parts $F^S(\rho)$ and $F^A$, which are related to the behaviour of
the system under time-reversal, as well as to the splitting of the heat current into ``excess'' and ``housekeeping''
contributions~\cite{Seifert2012a}. We use this splitting to decompose the function $\Phi$ into {three} pieces, which
allows us to compare (for example) the behaviour of reversible and irreversible Markov chains.  This splitting also mirrors a
similar construction within Macroscopic Fluctuation Theory~\cite{Bertini2015a}, and this link will be discussed in
Section~\ref{sec:Connections-to-MFT}. Related splittings have been introduced elsewhere; 
{see~\cite{Kwon2005a} and~\cite{qian2013} for decompositions of forces in stochastic differential equations}, and~\cite{Carlo2017a} for decompositions of the instantaneous current in interacting particle systems.

\subsection{Splitting of the force according to time-reversal symmetry}
\label{sec:Splitt-force-accord}

We define the \emph{adjoint process} associated with the original Markov chain of interest. The transition rates of the adjoint
process are $r^*_{xy}:=\pi(y)r_{yx}\pi(x)^{-1}$.  It is easily verified that the adjoint process has invariant measure $\pi$, so
$q^*_{xy}:=\pi(x)r^*_{xy}=q_{yx}$.  Under the assumption that the initial distribution is sampled from the steady state, the
probability to observe a trajectory for the adjoint process coincides with the probability to observe the time-reversed trajectory
for the original process.

From the definition of $F(\rho)$ in~\eqref{equ:def-aF}, we can decompose this force as
\begin{equation}
  F_{xy}(\rho) = F^S_{xy}(\rho) + F^A_{xy}
\end{equation}
with
\begin{equation}
  \label{equ:FrS}
  F^S_{xy}(\rho) := -\nabla^{x,y}\log\frac\rho\pi, \qquad F^A_{xy} := \log\frac{q_{xy}}{q_{yx}}. 
\end{equation}
With this choice, we note that the equivalent force for the adjoint process
\begin{equation*}
  F^*(\rho) = \log \frac{\rho(x) r^*_{xy} }{ \rho(y) r^*_{yx} },
\end{equation*}
satisfies $F^*(\rho)=F^S(\rho) - F^A$. So taking the adjoint inverts the sign of $F^A$ (the ``antisymmetric'' force) but leaves
$F^S(\rho)$ unchanged (the ``symmetric'' force).  For a reversible Markov chain, the adjoint process coincides with the original
one, and $F^A=0$.

\begin{lemma}
  \label{lem:one}
  Given $\rho$, with the mobility $a(\rho)$ of~\eqref{equ:def-aF}, the forces $F^S(\rho)$ and $F^A$ satisfy
  \begin{equation}
    \label{eqn:HJ} 
    \sum_{xy}\sinh\bigl(F^S_{xy}(\rho)/2\bigr)\;\! a_{x,y}(\rho) \sinh\bigl(F^A_{xy}/2\bigr)=0.
  \end{equation}
\end{lemma}

\begin{proof}
  From the definitions of $F^S(\rho)$, $F^A$, $a_{xy}$ and $\sinh$, one has
  \begin{equation*}
    a_{xy}(\rho)\sinh( F^S_{xy}(\rho)/2)=\Bigl(\frac{\rho(x)}{\pi(x)}-\frac{\rho(y)}{\pi(y)}\Bigr)\sqrt{q_{xy}q_{yx}}
  \end{equation*}
  and $\sinh (F^A_{xy}/2) = (q_{xy}q_{yx})^{-1/2}(q_{xy}-q_{yx})/2$.  Hence
\begin{multline*}
\qquad\sum_{xy}\sinh\bigl(F^S_{xy}(\rho)/2\bigr)\;\! a_{xy}(\rho) \sinh\bigl(F^A_{xy}/2\bigr) 
  = \frac 12 \sum_{xy} \Bigl(\frac{\rho(x)}{\pi(x)}-\frac{\rho(y)}{\pi(y)}\Bigr)(q_{xy}-q_{yx})\\
  = \sum_x \frac{\rho(x)}{\pi(x)}\sum_y(q_{xy}-q_{yx})=0,\qquad
\end{multline*}
where the last equality uses~\eqref{equ:q-balance}. This establishes~\eqref{eqn:HJ}.
\qed\end{proof}

In Section~\ref{sec:Decomp-force-F}, we will reformulate the so-called Hamilton-Jacobi relation of MFT in terms of forces, and
show that this yields an equation analogous to~\eqref{eqn:HJ}.

\subsection{Physical interpretation of \texorpdfstring{$F^S$}{symmetric force} and \texorpdfstring{$F^A$}{antisymmetric force}}
\label{sec:Phys-interpr-FS}

In stochastic thermodynamics, one may identify $F^A_{xy}$ as the \emph{housekeeping heat} (or \emph{adiabatic entropy production})
associated with a single transition from state $x$ to state $y$, see~\cite{Seifert2012a,Esposito2010a}. (Within the Markov chain
formalism, there is some mixing of the notions of force and energy: usually an energy would be a product of a force and a distance
but there is no notion of a distance between states of the Markov chain, so forces and energies have the same units in our
analysis.)  Hence $j\cdot F^A$ is the rate of flow of housekeeping heat into the environment.  The meaning of the housekeeping
heat is that for irreversible systems, transitions between states involve unavoidable dissipated heat which cannot be transformed
into work (this dissipation is required in order to ``do the housekeeping'').

To obtain the physical interpretation of $F^S$, we also define
\begin{equation}
  \label{eqn:free_energy_dissipation}
  D(\rho,j) :
  =\frac12 \sum_{xy} j_{xy} \log \frac{\rho(y)\pi(x)}{\rho(x)\pi(y)}.
\end{equation}
For a general path $(\rho_t,j_t)_{t\in[0,T]}$ that satisfies $\dot\rho_t = - \div j_t$, we also identify
\begin{equation}
  \label{eq:F-dot}
  \frac{d}{dt} \mathcal{F}(\rho_t) 
  = \sum_x \dot\rho_t(x) \log \frac{\rho_t(x)}{\pi(x)} = \frac 12\sum_{xy} (j_t)_{xy} \nabla^{x,y}\log\frac\rho\pi
  =  D(\rho_t,j_t) , 
\end{equation}
where we used~\eqref{eqn:free_energy},~\eqref{equ:parts}. That is, $D(\rho,j)$ is the change in free energy induced by the current
$j$.  Moreover it is easy to see that
\begin{equation}
  \label{eqn:FsGradF}
  F^S_{xy}(\rho) = -\nabla^{x,y} \frac{\delta \mathcal{F}}{\delta\rho} , 
\end{equation}
where $\frac{\delta\mathcal F}{\delta \rho}$ denotes the functional derivative of the free energy $\cal F$ {given
  in~\eqref{eqn:free_energy}. (Note that the functional derivative $\delta\mathcal F/\delta \rho$ is simply
  $\partial{\cal F}/\partial\rho$ in this case, since $\rho$ is defined on a discrete space. We retain the functional notation to
  emphasise the connection to the general setting of Section~\ref{sec:Summary}).} Also, the last identity in~\eqref{eq:F-dot} can
be phrased as
\begin{equation} 
  \label{eqn:JFS-dissipation} 
  j\cdot F^S(\rho) =  -D(\rho,j).
\end{equation}
The same identity, with an integration by parts, shows that 
\begin{equation}
  \label{eq:div-free-D-null}
  D(\rho, j) = 0 \text{ if $j$ is divergence free.} 
\end{equation}

Equation~\eqref{eqn:FsGradF} shows that the symmetric force $F^S$ is {minus} the gradient of the free energy, so the
heat flow associated with the dual pairing of $j$ and $F^S$ is equal to (the negative of) the rate of change of the free energy.
It follows that the right hand side of~\eqref{equ:gc-finite-time} can alternatively be written as
$-\int j\cdot F^A\, \mathrm{d}t$.

We also recall from Section~\ref{sec:Non-linear-flux} that the force $F$ acts in the space of probability densities: $F_{xy}$
depends not only on the states $x,y$ but also on the density $\rho$.  (Physical forces acting on individual copies of the system
should not depend on $\rho$ since each copy evolves independently, but $F$ includes entropic terms associated with the ensemble of
copies.) To understand this dependence, it is useful to write
$\mathcal{F}(\rho) = -\sum_x \rho(x) \log \pi(x) + \sum_x \rho(x) \log \rho(x)$.  We also write the invariant measure in a
Gibbs-Boltzmann form: $\pi(x) = \exp(-U(x))/Z$, where $U(x)$ is the internal energy of state $x$ and $Z=\sum_x \exp(-U(x))$ is a
normalisation constant.  Then $-\sum_x \rho(x) \log \pi(x) = \mathbb{E}_\rho(U) + \log Z$ depends on the mean energy of the
system, while $\sum_x \rho(x) \log \rho(x)$ is (the negative of) the mixing entropy, which comes from the many possible
permutations of the copies of the system among the states of the Markov chain.  From~\eqref{eqn:FsGradF} one then sees that $F^S$
has two contributions: one term (independent of $\rho$) that comes from the gradient of the energy $U$ and the other (which
depends on $\rho$) comes from the gradient of the entropy.  These entropic forces account for the fact that a given empirical
density $\rho^{\;\!\cN}$ can be achieved in many different ways, since individual copies of the system can be permuted among the
different states of the system.

\subsection{Generalised orthogonality for forces}
\label{sec:Decomp-rate-funct}

Recalling the definitions of Section~\ref{sec:Splitt-force-accord}, one sees that the current in the adjoint process satisfies an
analogue of~\eqref{eq:J-sinh}:
\begin{equation}
  \label{eqn:def_adj_current}
  J^\ast_{xy}(\rho) := a_{xy}(\rho) \sinh \bigl(\tfrac12 F^*_{xy}(\rho)\bigr), \qquad\text{with}\qquad
  F^*_{xy}(\rho):= F^S_{xy}(\rho) - F^A_{xy}.
\end{equation}
Comparing with~\eqref{equ:FrS}, one sees that the adjoint process may also be obtained by inverting $F^A$ (while keeping
$F^S(\rho)$ as it is).  For $a_{xy}^S(\rho):= a_{xy}(\rho)\cosh(F^A_{xy}/2)$ the symmetric current is defined as
\begin{equation}
  \label{eqn:J_S}
  J^S_{xy}(\rho) : = a_{xy}^S(\rho)\sinh\bigl( F^S_{xy}(\rho)/2\bigr),
\end{equation}
which satisfies $J^S_{xy}(\rho) = (J_{xy}(\rho) + J^\ast_{xy}(\rho))/2$. It is the same for the process and the adjoint process,
and also coincides with the current for reversible processes (where $q_{xy}=q_{yx}$, or equivalently $F^A=0$).  {An
  analogous formula can also be obtained for the anti-symmetric current. With
  $a^A_{xy}(\rho) := a_{xy}(\rho)\cosh(F^S_{xy}(\rho)/2)=a_{xy}(\pi)\bigl(\frac{\rho(x)}{\pi(x)}+\frac{\rho(y)}{\pi(y)}\bigr)/2$,
  the anti-symmetric current is defined as
\begin{equation}
J^A_{xy}(\rho) := a^A_{xy}(\rho)\sinh\bigl(F^A_{xy}/2\bigr).
\end{equation}
It satisfies $J^A_{xy}(\rho) = (J_{xy}(\rho) - J^\ast_{xy}(\rho))/2$.
}

Let $\Psi_S^\star$ be the symmetric version of $\Psi^\star$ obtained from~\eqref{eqn:psi_star} with $a_{xy}(\rho)$ replaced by
$a^S_{xy}(\rho)$. (The Legendre transform of $\Psi^\star_S$ is similarly denoted $\Psi_S$). This leads to a separation of
$\Psi^\star(\rho,F(\rho))$ in a term corresponding to $F^S(\rho)$ and a term corresponding to $F^A$.

\begin{lemma} 
  \label{lem:psi_split} 
  The two forces $F^S(\rho)$ and $F^A$ defined in~\eqref{equ:FrS} satisfy
  \begin{equation}
    \label{eqn:new_psi_star_2}
    \Psi^\star(\rho,F(\rho))
    =\Psi_S^\star\bigl(\rho,F^S(\rho)\bigr) + \Psi^\star\bigl(\rho,F^A\bigr),
  \end{equation}
\end{lemma}

\begin{proof}
  Using $\cosh(x+y) = \cosh(x)\cosh(y) + \sinh(x)\sinh(y)$, Lemma~\ref{lem:one} and the definition of $a_{xy}^S(\rho)$, we obtain
  that the left hand side of~\eqref{eqn:new_psi_star_2} is given by
  \begin{multline}
    \sum_{xy} a_{xy}(\rho)\bigl(\cosh(F_{xy}(\rho)/2)-1\bigr) 
    = \sum_{xy} a_{xy}(\rho)\bigl(\cosh(F^S_{xy}(\rho)/2)\cosh(F^A_{xy}(\rho)/2)-1\bigr)\\
    = \sum_{xy} a_{xy}^S(\rho)\bigl(\cosh(F^S_{xy}(\rho)/2) - 1\bigr) + \sum_{xy} a_{xy}(\rho)\bigl(\cosh(F^A_{xy}(\rho)/2) - 1\bigr),
  \end{multline}
  which coincides with the right hand side of~\eqref{eqn:new_psi_star_2}.  \qed
\end{proof} 
The physical interpretation of Lemma~\ref{lem:psi_split} is that the strength of the force $F(\rho)$ can be written as separate
contributions from $F^S(\rho)$ and $F^A$. The following corollary allows us to think of a generalised orthogonality of the forces
$F^S(\rho)$ and $F^A$.

\begin{proposition}[Generalised orthogonality]
  \label{prop:orth}
  The forces $F^S(\rho)$ and $F^A$ satisfy
  \begin{equation}
    \label{eqn:orth}
    \Psi^\star\bigl(\rho,F^S(\rho)+F^A\bigr) = \Psi^\star\bigl(\rho, F^S(\rho)-F^A\bigr).
  \end{equation}
\end{proposition}

\begin{proof}
  This follows directly from Lemma~\ref{lem:psi_split} and the symmetry of $\Psi^\star(\rho,\cdot)$.
\qed
\end{proof}

We refer to Proposition~\ref{prop:orth} as a generalised orthogonality between $F^S$ and $F^A$ because $\Psi^\star$ is acting as
generalisation of a squared norm (see Section~\ref{sec:Summary}), so~\eqref{eqn:orth} can be viewed as a nonlinear generalisation
of $\| F^S + F^A \|^2 = \| F^S - F^A \|^2$, which would be a standard orthogonality between forces.

Moreover, Lemma~\ref{lem:psi_split} can be used to decompose the OM functional as a sum of three terms.
\begin{corollary}
  \label{cor:two}
  Let $\Phi_S$ be defined as in~\eqref{eqn:Phi_function} with $(\Psi,\Psi^\star)$ replaced by $(\Psi_S,\Psi_S^\star)$, and
  $D(\rho,j)$ as defined in~\eqref{eqn:free_energy_dissipation}. Then
  \begin{equation}
    \label{eqn:lagrangian_2}
    \Phi(\rho,j,F(\rho)) =D(\rho,j)  + \Phi_S\bigl(\rho,0,F^S(\rho)\bigr) + \Phi\bigl(\rho,j,F^A\bigr).
  \end{equation}
\end{corollary}

\begin{proof}
  We use the definition of $\Phi$ in~\eqref{eqn:Phi_function} and~\eqref{eqn:JFS-dissipation} together with
  Lemma~\ref{lem:psi_split} to decompose $\Phi(\rho,j,F(\rho))$ as
\begin{equation}
  \begin{split}
    \Phi(\rho,j,F(\rho)) &= D(\rho,j) +\Psi_S^\star\bigl(\rho,F^S(\rho)\bigr) + \Bigl[ \Psi(\rho,j) - j\cdot  F^A
      +\Psi^\star\bigl(\rho,F^A\bigr)\Bigr] \\
    &=D(\rho,j)  + \Phi_S\bigl(\rho,0,F^S(\rho)\bigr) + \Phi\bigl(\rho,j,F^A\bigr),
  \end{split}
\end{equation}
which proves the claim. 
\qed
\end{proof}

Recall from Section~\ref{sec:Summary} that $\Phi$ measures how much the current $j$ deviates from the typical (or most likely)
current $J(\rho)$.  One sees from~\eqref{eqn:lagrangian_2} that it can be large for three reasons.  The first term is large if the
current is pushing the system up in free energy (because $D$ is the rate of change of free energy induced by the current $j$).
The second term comes from the time-reversal symmetric (gradient) force $F^S(\rho)$, which is pushing the system towards
equilibrium.  The third term comes from the time-reversal anti-symmetric force $F^A$; namely, it measures how far the current $j$
is from the value induced by the force $F^A$.

Corollary~\ref{cor:two} also makes it apparent that the free energy $\mathcal F$ is monotonically decreasing for solutions
of~\eqref{eq:master}, which are minimisers of $I_{[0,T]}$.

\begin{corollary} 
  The free energy $\mathcal F$ is monotonically decreasing along minimisers of the rate function $I_{[0,T]}$. Its rate of change
  is given by
  \begin{equation}
    \frac d{dt} \mathcal F(\rho_t) = -\Psi^\star_S\bigl(\rho_t,F^S(\rho_t)\bigr) - \Phi\bigl(\rho_t,J(\rho_t),F^A(\rho_t)\bigr).
  \end{equation} 
\end{corollary}

\begin{proof}
  For minimisers of the rate function one has $\Phi=0$.  Hence~\eqref{eq:F-dot} and Corollary~\ref{cor:two} imply that
  \begin{equation}
    \label{eq:F-dot-2}
    \frac d{dt} \mathcal F(\rho_t) 
    = D(\rho,j) 
    = -\Psi^\star_S\bigl(\rho_t,F^S(\rho_t)\bigr) - 
    \Phi\bigl(\rho_t,J(\rho_t),F^A(\rho_t)\bigr).
  \end{equation}
  {Both $\Psi^\star$ and $\Phi$ are non-negative, so $\cal F$ is indeed monotonically decreasing.}
\qed
\end{proof}

\subsection{{Hamilton-Jacobi like equation} for Markov chains}
\label{sec:HJ-for-MC}

It is also useful to note at this point an additional aspect of the orthogonality relationships presented here, which has
connections to MFT (see Section~\ref{sec:Connections-to-MFT}).  We formulate an analogue of the Hamilton-Jacobi equation of MFT,
as follows.  Define
\begin{equation}
  \mathbb{H}(\rho,\xi) = \frac12\left[ \Psi^\star(\rho,F(\rho) + 2\xi) -   \Psi^\star(\rho,F(\rho)) \right] , 
  \label{equ:ham-markov}
\end{equation}
which we refer to as an {\emph{extended Hamiltonian}}, for reasons discussed in Section~\ref{sec:Hamilt-Jacobi-equat} below {(see also Section~IV.G of~\cite{Bertini2015a})}.

{
The {\it extended Hamilton-Jacobi equation} for a functional $\mathcal S$ is then (cf.~equation \eqref{eqn:HJ_micro} in Section~\ref{sec:Hamilt-Jacobi-equat}) given by 
\begin{equation}
  \label{equ:HJ-markov}
  \mathbb H\left(\rho,\nabla\frac{\delta \mathcal S}{\delta\rho}\right)=0.
\end{equation}
Note that the free energy $\cal F$ defined in~\eqref{eqn:free_energy} solves \eqref{equ:HJ-markov},}
which follows from Proposition~\ref{prop:orth} (using~\eqref{eqn:FsGradF} and that $\Psi^\star$ is symmetric in its second
argument).  In fact (see Proposition~\ref{prop:HJ}), the free energy is the maximal solution to this equation. In MFT, the
analogous variational principle can be useful, as a characterisation of the invariant measure of the process.  Here, one has a
similar characterisation of the (non-equilibrium) free energy.

Since~\eqref{equ:HJ-markov} {with $\mathcal S = \cal F$} provides a characterisation of the free energy $\cal F$, which is uniquely determined by the invariant
measure $\pi$ of the process, it follows that~\eqref{equ:HJ-markov} must be equivalent to the condition that $\pi$ satisfies
$\div J(\pi)=0$: recall~\eqref{eq:master}.  Writing everything in terms of the rates of the Markov chain and its adjoint,
\eqref{equ:HJ-markov} becomes
\begin{equation*}
  \sum_x \rho(x) \sum_{y} [r_{xy} - r^*_{xy}] = 0 ,
\end{equation*}
which must hold for all $\rho$: from the definition of $r^*$ one then has $\sum_y \pi(x)r_{xy}=\sum_y \pi(y)r_{yx}$, which is
indeed satisfied if and only if $\pi$ is invariant (cf. equation~\eqref{equ:q-balance}).

\begin{figure}
  \begin{center}\includegraphics[width=5cm]{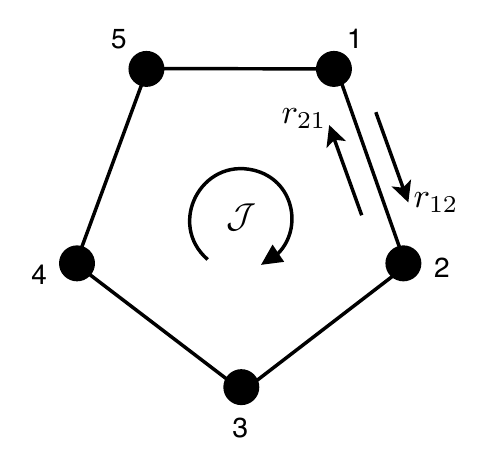}\end{center}
  \caption{ Illustration of a simple Markov chain with $n=5$ states arranged in a circle.  The transition rates
    between states are $r_{i,i\pm1}$.  If the Markov chain is not reversible, there will be a steady-state probability current
    ${\cal J}$ corresponding to a net drift of the system around the circle.}
  \label{fig:circle}
\end{figure}

{

\subsection{Example: simple ring network}
\label{sec:ring}

To illustrate these abstract ideas, we consider a very simple Markov chain, in which $n$ states are arranged in a circle, see
Fig.~\ref{fig:circle}. So $V=\{1,2,\dots,n\}$ and the only allowed transitions take place between state $x$ and states $x\pm 1$
(to incorporate the circular geometry we interpret $n+1=1$ and $1-1=n$).  In physics, such Markov chains arise (for example) as
simple models of nano-machines or motors, where an external energy source might be used to drive circular
motion~\cite{fisher07,vaik14}.  Alternatively, such a Markov chain might describe a protein molecule that goes through a cyclic
sequence of conformations, as it catalyses a chemical reaction~\cite{lavorel76}.  In both cases, the systems evolve stochastically
because the relevant objects have sizes on the nano-scale, so thermal fluctuations play an important role.

To apply the analysis presented here, the first step is to identify forces and mobilities, as in~\eqref{equ:def-aF}.  Let
$R_x = \sqrt{r_{x,x+1} r_{x+1,x}}$.  The invariant measure may be identified by solving
$\sum_y \pi(x) r_{xy}= \sum_y \pi(y) r_{yx}$ subject to $\sum_y \pi(y)=1$. Finally, one computes the steady state current
${\cal J} = \pi(x) r_{x,x+1} - \pi(x\!+\!1) r_{x+1,x}$, where the right hand side is independent of $x$ (this follows from the
steady-state condition on $\pi$).  The original Markov process has $2n$ parameters, which are the rates $r_{x,x\pm 1}$: these are
completely determined by the $n-1$ independent elements of $\pi$, the $n$ mobilities $(R_x)_{x=1}^n$ and the current $\cal J$.
The idea is that this reparameterisation allows access to the physically important quantities in the system.

From the definitions of $\cal J$ and $R$, it may be verified that
\begin{equation*}
  2 \pi(x) r_{x,x+1} = \sqrt{{\cal J}^2 + 4 R_x^2 \pi(x)\pi(x\!+\!1)} + {\cal J},
\end{equation*}
and similarly $2 \pi(x\!+\!1) r_{x+1,x} = \sqrt{{\cal J}^2 + 4 R_x^2 \pi(x)\pi(x\!+\!1)} - {\cal J}.$ Then write
\begin{equation}
\rho(x) r_{x,x+1}  = R_x \sqrt{\rho(x)\rho(x+1)} \times \sqrt{ \frac{\rho(x)\pi(x\!+\!1)}{\rho(x\!+\!1)\pi(x)} } 
   \times \left( \frac{\sqrt{{\cal J}^2 + 4 R_x^2 \pi(x)\pi(x\!+\!1)} + {\cal J}}{ \sqrt{{\cal J}^2 + 4 R_x^2 \pi(x)\pi(x\!+\!1)} 
      - {\cal J} } \right)^{1/2}.
\end{equation}
In this case, we can identify the three terms as
\begin{equation}
  \rho(x) r_{x,x+1} 
  = \frac12 a_{x,x+1}(\rho) \times \exp(F^S_{x,x+1}(\rho)/2) \times  \exp(F^A_{x,x+1}/2) ,
\end{equation}
which allows us to read off the mobility $a$ and the forces $F^S$ and $F^A$.  The physical meaning of these quantities may not be
obvious from these definitions, but we show in the following that reparameterising the transition rates in this way reveals
structure in the dynamical fluctuations.

For example, equilibrium models (with detailed balance) can be identified via $F^A_{x,x+1}=0$ (for all $x$).  In general
$F^A_{x,x+1}$ is the (steady-state) entropy production associated with a transition from $x$ to $x+1$, see
Section~\ref{sec:Phys-interpr-FS}.  The steady state entropy production associated with going once round the circuit is
$\sum_x F^A_{x,x+1}=\log \prod_x (r_{x,x+1}/r_{x+1,x})$, as it must be~\cite{Andrieux2007a}.

Now consider the LDP in~\eqref{eqn:ldp_pathwise_statement}.  We consider a large number ($\cal N$) of identical nano-scale
devices, each of which is described by an independent copy of the Markov chain.  Typically, each device goes around the circle at
random, and the average current is ${\cal J}$ (so each object performs ${\cal J}/n$ cycles per unit time).  The LDP describes
properties of the ensemble of devices.  If $\cal N$ is large and the distribution of devices over states is $\rho$, then the
(overwhelmingly likely) time evolution of this distribution is $\dot\rho = -\div J(\rho)$, where the current $J$ obeys the simple
formula
\begin{equation}
J_{x,x+1}(\rho) = a_{x,x+1}(\rho) \sinh\left( \tfrac12 [F^S_{x,x+1}(\rho) + F^A_{x,x+1}] \right) ,
\end{equation}
which is~\eqref{eq:J-sinh}, applied to this system.  The simplicity of this expression motivates the parametrisation of the
transition rates in terms of forces and mobilities.  In addition, if one observes some current $j$ [not necessarily equal to
$J(\rho)$] then the rate of change of free energy of the ensemble can be written compactly as $D(\rho,j) = -j\cdot F^S(\rho)$,
from~\eqref{eqn:JFS-dissipation}.  The quantity $j\cdot F^A$ is the rate of dissipation via housekeeping heat (see
Section~\ref{sec:Phys-interpr-FS}).  This (physically-motivated) splitting of $j\cdot F=j\cdot (F^S+F^A)$ motivates our
introduction of the two forces $F^S$ and $F^A$.  Note that $j \cdot F$ is the rate of heat flow from the system to its environment,
and appears in the fluctuation theorem~\eqref{equ:gc-finite-time}.

Finally we turn to the large deviations of this ensemble of nano-scale objects.  There is an
LDP~\eqref{eqn:ldp_pathwise_statement}, whose rate function can be decomposed into three pieces (Corollary~\ref{cor:two}), because
of the generalised orthogonality of the forces $F^S$ and $F^A$ (Lemma~\ref{lem:psi_split}).  This splitting of the rate function
is useful because the symmetry properties of the various terms yields bounds on rate functions for some other LDPs obtained from
$\Phi$ by contraction, see Section~\ref{sec:LDPs-time-averaged} below.

}

\section{Connections to MFT}
\label{sec:Connections-to-MFT}

Macroscopic Fluctuation Theory (MFT) is a field theory which describes the mass evolution of particle systems in the
drift-diffusive regime, on the level of hydrodynamics. In this setting, it can be seen as generalisation of Onsager-Machlup
theory~\cite{Machlup1953a}. For a comprehensive review, we refer to~\cite{Bertini2015a}. This section gives an overview of the
theory, {focussing on} the connections to the results presented in Sections~\ref{sec:Onsag-Machl-theory-Markov}
and~\ref{sec:Decomp-forc-rate}.

{We seek to emphasise two points: first, while the particle currents in MFT and the probability current in Markov
  chains are very different objects, they both obey large-deviation principles of the form presented in Section~\ref{sec:Summary}.
  This illustrates the broad applicability of this general setting.  Second, we note that many of} the particle models for which
MFT gives a macroscopic description are Markov chains on discrete spaces.  {Starting from this observation, we argue
  in Section~\ref{sec:hydro} that some results that are well-known in MFT originate from properties of these underlying Markov
  chains, particularly Proposition~\ref{prop:orth} and~Corollary~\ref{cor:two}.}

\subsection{Setting}
\label{sec:Setting-1}

We consider a large number $N$ of indistinguishable particles, moving on a lattice $\Lambda_L$ (indexed by $L\in\N$, such that the
number of sites $|\Lambda_L|$ is strictly increasing with $L$).  These particles are described by a Markov chain, so the relevant
forces and currents satisfy the equations derived in Sections~\ref{sec:Onsag-Machl-theory-Markov}
and~\ref{sec:Decomp-forc-rate}. The hydrodynamic limit is obtained by letting $L\to\infty$ such that the total density
$N/|\Lambda_L|$ converges to a fixed number $\bar\rho$. In this limit, the lattice $\Lambda_L$ is rescaled into a domain
$\Lambda\subset \R^d$ and one can characterise the system by a local (mass) density $\rho \colon \Lambda\to[0,\infty)$ together
with a local current $j \colon \Lambda\to\R^d$, which evolve deterministically as a function of
time~\cite{Kipnis1999a,Bertini2015a}.  This time evolution depends on some (density-dependent) applied forces
$F(\rho) \colon \Lambda\to\R^d$. The force at $x\in\Lambda$ can be written as
\begin{equation}
  \label{equ:Fmft}
  { F(\rho)(x)=\hat{f}''(\rho(x))\nabla \rho(x) + E(x),}
\end{equation} 
where the gradient $\nabla$ denotes a spatial derivative, the function $\hat{f} \colon [0,\infty)\to\mathbb{R}$ is a free energy
density and $E\colon \Lambda\to\R^d$ is a drift.  {(The free energy $\hat{f}$ is conventionally denoted by
  $f$~\cite{Bertini2015a}; here we use a different notation since $f$ indicates a force in this work.)}  With these definitions,
the deterministic currents satisfy the linear relation~\cite{Maes2015a}
\begin{equation}
  \label{equ:mft-Jrho} 
  J(\rho)=\chi(\rho) F(\rho) ,
\end{equation}
which is the hydrodynamic analogue of~\eqref{eq:J-sinh}.  Here, $\chi(\rho) \in \mathbb{R}^{d \times d}$ is a (density-dependent)
mobility matrix.

\subsection{Onsager-Machlup functional}
\label{sec:Onsag-Machl-funct}
 
Within MFT, the system is fully specified once the functions $f,\chi,E$ are given. These three quantities are sufficient to
specify both the deterministic evolution of the most likely path $\rho$, and the fluctuations away from it. We can again define an
OM functional given by
\begin{equation}
  \label{eqn:psi_mft}
  \Phi_{\MFT}(\rho,j,f) :=\frac 12\int_\Lambda \bigl(j-\chi f\bigr)\cdot\chi^{-1}\bigl(j - \chi f\bigr) \df x.
\end{equation}
To cast this functional in the form~\eqref{eqn:Phi_function}, we define the dual pair $\int_\Lambda (j\cdot f) \df x$, together
with the Legendre duals
\begin{equation}
  \label{eq:MFT-Psi}
  \Psi_{\MFT}(\rho,j):=\frac 12\int_\Lambda j\cdot\chi^{-1}j \df x \quad\text{and}\quad 
  \Psi^\star_{\MFT}(\rho,f):=\frac 12\int_\Lambda f\cdot\chi f \df x . 
\end{equation}
Given $\rho$ and $f$, we have that $\Phi_{\MFT}$ is uniquely minimised (and equal to zero) for the current $j = \chi(\rho) f$.

\subsection{Large deviation principle}
\label{sec:Large-devi-princ}

{ Within MFT, one considers an empirical density and an empirical current. We emphasise that these
  refer to particles, which are interacting and move on the lattice $\Lambda_L$; this is in contrast to the case of Markov chains,
  where the copies of the system were non-interacting and one considers a density and current of probability.  The averaged number of
  particles at site $i\in\Lambda_L$ is denoted with $\hat\rho_t^{\lL}(x_i)$, where $x_i$ is the image in the rescaled domain
  $\Lambda$ of site $i\in\Lambda_L$, and the} corresponding particle current is given by $\hat\jmath_t^{\lL}$ {(cf.~Section~VIII.F~in \cite{Bertini2015a} for details).}  Note that both the
particle density $\hat\rho_t^{\lL}$ and the particle current $\hat\jmath_t^{\lL}$ are random quantities {(see also
  Section~\ref{sec:hydro} below).

  In keeping with the setting of Section~\ref{sec:Summary}}, we focus on paths
$(\hat\rho_t^{\!\;L},\hat\jmath_t^{\lL})_{t\in[0,T]}$ in the limit as $L\to\infty$, where the probability is, analogous
to~\eqref{equ:pathwise-general}, given by
\begin{equation}
  \label{eqn:pathwise_ldp_mft}
  \mathrm{Prob}\Bigl( (\hat\rho_t^{L},\hat\jmath_t^{\lL})_{t\in[0,T]}\approx
  (\rho_t, j_t)_{t\in[0,T]}\Bigr) \asymp \exp\bigl\{-|\Lambda_L|  I_{[0,T]}^{\MFT}\bigl( (\rho_t, j_t)_{t\in[0,T]}\bigr)\bigr\}. 
\end{equation}
 {Note that the parameter $\cN$ in~\eqref{equ:pathwise-general},
  which is the speed of the LDP, corresponds to the lattice size $|\Lambda_L|$}. For the force $F(\rho)$ defined in~\eqref{equ:Fmft}, the rate functional
in~\eqref{eqn:pathwise_ldp_mft} is given by
\begin{equation}
  \label{equ:pathwise-mft}
  I_{[0,T]}^{\mathrm{MFT}}\bigl( (\rho_t, j_t)_{t\in[0,T]}\bigr)\! =\!
  \begin{cases}
    \mathcal V(\rho_0)\! +\! \frac12\! \int_0^T\!\Phi_{\MFT}(\rho_t,j_t,F(\rho_t)) \df t & \text{if }\dot\rho_t\!+\!\div j_t\! =\! 0\\
    +\infty & \text{otherwise}.
  \end{cases}
\end{equation}
Here $\cal V$ is the \emph{quasipotential}, which plays the role of a non-equilibrium free energy. We may think of $\mathcal V$ as
the macroscopic analogue of the free energy $\mathcal F$ defined in~\eqref{eqn:free_energy}.  It is the rate functional for the
process sampled from the invariant measure, which is consistent with the case for Markov chains in~\eqref{eqn:ldp_pathwise}.  We
assume that $\mathcal V$ has a unique minimiser $\pi$, which is the steady-state density profile (so $\mathcal V(\pi)=0$).

An important difference between the Markov chain setting and MFT is that the OM functional for Markov chains is non-quadratic,
which is equivalent to a non-linear flux force relation, whereas MFT is restricted to quadratic OM functionals.

Equation~\eqref{eqn:pathwise_ldp_mft} is the basic assumption in MFT~\cite{Bertini2015a}, in the sense that all systems considered
by MFT are assumed to satisfy this pathwise LDP. In fact, both the process and its adjoint are assumed to satisfy such LDPs (with
similar rate functionals, but different forces)~\cite{Bertini2015a}.

\subsection{Decomposition of the force $F$}
\label{sec:Decomp-force-F}

The force $F$ in~\eqref{equ:Fmft} can be written as the sum of a symmetric and an anti-symmetric part,
$F(\rho)=F_S(\rho)+F_A(\rho)$, just as in Section~\ref{sec:Splitt-force-accord}. The force for the adjoint process is given by
$F^\ast(\rho)=F_S(\rho)-F_A(\rho)$. Note that, unlike in the case of Markov chains, $F_A(\rho)$ can here depend on $\rho$.  More
precisely, $F_S(\rho) = -\nabla\frac{\delta\mathcal V}{\delta \rho}$ and $F_A(\rho)$ is given implicitly by
$F_A(\rho) = F(\rho)-F_S(\rho)$.

The symmetric and anti-symmetric currents are defined in terms of the forces $F_S(\rho)$ and $F_A(\rho)$ as
$J_S(\rho) := \chi(\rho)F_S(\rho)$ and $J_A(\rho) := \chi(\rho)F_A(\rho)$.  An important result in MFT is the so-called
\emph{Hamilton-Jacobi orthogonality}, which states that
\begin{equation}
  \label{eqn:HJ-MFT}
  \int_\Lambda J_S(\rho)\cdot\chi(\rho)^{-1} J_A(\rho) \df x = 0.
\end{equation}
In terms of the forces $F_S(\rho)$ and $F_A(\rho)$, we can restate~\eqref{eqn:HJ-MFT} as
\begin{equation}
  \label{eq:HF-MFT-force}
  \int_\Lambda F_S(\rho) \cdot\chi(\rho) F_A(\rho) \df x = 0 . 
\end{equation}
The latter is the quadratic version of the orthogonality~\eqref{eqn:HJ} of Lemma~\ref{lem:one}; it is equivalent to
\begin{equation}
  \int_\Lambda \bigl(F_S(\rho) +  F_A(\rho)\bigr) \cdot\chi(\rho) \bigl(F_S(\rho) +  F_A(\rho)\bigr)\df x
  = \int_\Lambda \bigl(F_S(\rho) - F_A(\rho)\bigr) \cdot\chi(\rho) \bigl(F_S(\rho) -  F_A(\rho)\bigr)\df x,
\end{equation}
or in other words, from~\eqref{eq:MFT-Psi},
\begin{equation}
  \label{eqn:generalised_HJ}
  \Psi^\star_{\MFT}(\rho,F_S(\rho)+F_A(\rho))=\Psi^\star_{\MFT}(\rho,F_S(\rho)-F_A(\rho)),
\end{equation}
which is the result of Proposition~\ref{prop:orth} in the context of MFT.  One can see~\eqref{eqn:orth}, and hence
Proposition~\ref{prop:orth}, as the natural generalisation to the Hamilton-Jacobi orthogonality~\eqref{eqn:HJ-MFT}. Again, the MFT
describes systems on the macroscopic scale, but the result~\eqref{eqn:generalised_HJ} originates from the result~\eqref{eqn:orth},
on the microscopic level.

{

\subsection{Relating Markov chains to MFT: hydrodynamic limits}
\label{sec:hydro}

We have discussed a formal analogy between current/density fluctuations in Markov chains and in MFT: the large deviation
principles~\eqref{eqn:ldp_pathwise_statement} and~\eqref{eqn:pathwise_ldp_mft} refer to different objects and different limits,
but they both fall within the general setting described in Section~\ref{sec:Summary}.  We argue here that the similarities between
these two large deviation principles are not coincidental -- they arise naturally when MFT is interpreted as a theory for
hydrodynamic limits of interacting particle systems.

To avoid confusion between particle densities and probability densities, we introduce (only for this section) a different notation
for some properties of discrete Markov chains, which is standard for interacting particle systems.  Let $\eta$ represent a state
of the Markov chain (in place of the notation $x$ of Section~\ref{sec:Onsag-Machl-theory-Markov}), and let $\mu$ be a probability
distribution over these states (in place of the notation $\rho$ of Section~\ref{sec:Onsag-Machl-theory-Markov}).  Let $\jmath$ be
the probability current.

We illustrate our argument using the weakly asymmetric {simple} exclusion process (WASEP) in one dimension, so the lattice is
$\Lambda_L=\{1,2,\dots,L\}$, and each lattice site contains at most one particle, so $V=\{0,1\}^L$. The lattice has periodic
boundary conditions and the occupancy of site $i$ is $\eta(i)$.  Particles hop to the right with rate $L^2$ and to the left with
rate $L^2(1-({ E}/L))$, but in either case only if the destination site is empty.  Here $ E$ is a fixed parameter
(an external field); the dependence of the hop rates on $L$ is chosen to ensure a diffusive hydrodynamic limit (as required for
MFT).

The spatial domain relevant for MFT is $\Lambda=[0,1]$: site $i\in\Lambda_L$ corresponds to position $i/L\in\Lambda$. For any
probability measure $\mu$ on $V$, one can write a corresponding smoothed particle density $\rho^\eps$ on $\Lambda$, as
\begin{equation}
  \rho^\eps(x) = \frac1L \sum_{\eta\in V} \sum_{i=1}^L \mu(\eta) \eta(i/L) \delta^\eps(x-(i/L)),
\end{equation}
where $\delta^\eps$ is a smoothed delta function (for example a Gaussian with unit weight and width $\eps$, or -- more classically
-- a top-hat function of width $\epsilon$, cf.~\cite{Kipnis1999a}).  Similarly if there is a probability current $\jmath$ in the
Markov chain, one can write a smoothed particle current as
 \begin{equation}
   j^\eps(x) = \frac1L  \sum_{\eta\in V} \sum_{i=1}^L \jmath_{\eta,\eta^{i,i+1}} \delta^\eps\left(x-\frac{2i+1}{2L} \right),
\end{equation}
where $\eta^{i,i+1}$ is the configuration obtained from $\eta$ by moving a particle from site $i$ to site $i+1$; if there is no
particle on site $i$ then define $\eta^{i,i+1}=\eta$ so that $\jmath_{\eta,\eta^{i,i+1}}=0$.  Physically, $\rho^\eps$ is the
average particle density associated to $\mu$, and $j^\eps$ is the particle current associated to $\jmath$.

As noted above, MFT is concerned with the limit $L\to\infty$.  The LDP~\eqref{eqn:ldp_pathwise_statement} is not relevant for that
limit (it applies when one considers many (${\cal N}\to\infty$) independent copies of the Markov chain, with $L$ being finite for
each copy).  However, the rate function $I_{[0,T]}$ that appears in~\eqref{eqn:ldp_pathwise_statement} has an alternative physical
interpretation, as the relative entropy between two path measures: see Appendix~\ref{sec:relent}. This relative entropy can be seen as a
property of the WASEP; there is no requirement to invoke many copies of the system. Physically, the relative entropy measures how
different is the WASEP from an alternative Markov process with a given probability and current $(\mu_t,\jmath_t)_{t\in[0,T]}$.

The key point is that in cases where MFT applies, one expects that the rate function $I^{\rm MFT}_{[0,T]}$ can be related to this relative entropy. 
{In fact, there is a deeper relation between relative entropies and rate functionals: it can be shown that Large Deviation Principles are equivalent to $\varGamma$-convergence of relative entropy functionals (see \cite{Mariani2012} for details).

Returning to the WASEP, we consider a particle density $(\rho_t,j_t)_{t\in[0,T]}$ that satisfies $\dot\rho_t = -\div j_t$. One then can find (for each $L$)} a time-dependent probability and current $(\mu_t^L,\jmath_t^L)_{t\in[0,T]}$, with $ \dot\mu^L_t = -\div \jmath^L_t$, {such on taking the limit $\epsilon\to0$ {\it after} $L\to\infty$}, the associated
particle densities $(\rho^\eps_t,j^\eps_t) \to (\rho_t,j_t)$ {and moreover
\begin{equation}
  \label{equ:mft-connect}
  \lim_{L\to\infty}\frac{1}{|\Lambda_L|} I_{[0,T]}\bigl( (\mu_t^L, \jmath_t^L)_{t\in[0,T]}\bigr) = I_{[0,T]}^{\mathrm{MFT}}
  \bigl( (\rho_t, j_t)_{t\in[0,T]}\bigr).
\end{equation}

In order to find $(\mu_t^L,\jmath_t^L)_{t\in[0,T]}$, one defines a ``controlled'' WASEP (similar to \eqref{equ:tilde-r} in Section~\ref{sec:Optim-contr-theory}), in which the particle hop rates depend on position and time, such that the particle density in the hydrodynamic limit obeys $\dot\rho_t = -\div j_t$. 

For interacting particle systems, this ``controlled'' process is usually obtained by adding a time dependent external field to the system that acts on the individual particles. This was first derived for the symmetric SEP in \cite{Kipnis1989} (see also \cite{Benois1995} for a treatment of the zero-range process). 
For the WASEP (in a slightly different situation with open boundaries) a proof of \eqref{equ:mft-connect} can e.g.~be found in \cite{Bertini2009}, Lemma 3.7.
}

Moreover, on decomposing $I^{\rm MFT}_{[0,T]}$ and $I_{[0,T]}$ as in~\eqref{eqn:Phi_function}, the separate functions $\Psi$ and
$\Psi^\star$ obey formulae analogous to~\eqref{equ:mft-connect}: this is the sense in which the structure of the MFT rate function
is inherited from the relative entropy of the Markov chains.  The quadratic functions $\Psi$ and $\Psi^\star$ in MFT arise because
the forces that appear in the underlying Markov chains are small (compared to unity), so second order Taylor
  expansions of $\Psi^\star$ and $\Psi$ give in the limit the accurate description.}
{We will return to this discussion in a later publication. } 

\section{LDPs for time-averaged quantities}
\label{sec:LDPs-time-averaged}

So far we have considered large deviation principles for hydrodynamic limits, and for systems consisting of many independent
copies of a single Markov chain.  We now show how some of the results derived in Sections~\ref{sec:Onsag-Machl-theory-Markov}
and~\ref{sec:Decomp-forc-rate} also have analogues for large deviations for a single Markov chain, in the large-time limit.

\subsection{Large deviations at level 2.5}
\label{sec:Large-deviations-at}

Analogous to~\eqref{eqn:N_average}, we define the time averaged {empirical measure of a single copy of the Markov chain} $\hat\rho_{[0,T]}$ and the time
averaged empirical current $\hat\jmath_{[0,T]}$ as
\begin{equation}
  \label{eqn:defn_time_averages}
  \hat\rho_{[0,T]}:=\frac 1T \int_0^T \hat\rho_t \df t
  \quad\text{and}\quad
  \hat\jmath_{[0,T]}:=\frac 1T\int_0^T\hat\jmath_t \df t
\end{equation}
{ (where we choose $\hat\rho_t=\hat\rho_t^1$ and $\hat\jmath_t = \hat\jmath_t^1$ for the empirical density and current of the single Markov chain, as defined above in Section~\ref{sec:Large-Devi-Onsag}).}
For countable state Markov chains, the quantity $(\hat\rho_{[0,T]},\hat\jmath_{[0,T]})$ satisfies a LDP as $T\to\infty$:
\begin{equation}
  \label{equ:l2.5-general}
  \mathrm{Prob}\bigl((\hat\rho_{[0,T]}, \hat\jmath_{[0,T]}) \approx (\rho,j)\bigr) \asymp \exp\bigl\{-T I_{2.5}(\rho,j)\bigr\} .
\end{equation}
We refer to such principles as \emph{level 2.5 LDPs}. For countable state Markov chains the rate functional $I_{2.5}(\rho,j)$ was
derived in~\cite{Maes2008b}, and was proven rigorously in~\cite{Bertini2015c,Bertini2015b} for Markov chains in the setting of
Section~\ref{sec:Setting} under some additional conditions (see~\cite{Bertini2015c,Bertini2015b} for the details). We can recast
the rate functional (see~\cite[Theorem 6.1]{Bertini2015c}) as
\begin{equation}
  \label{eqn:ldp_2.5}
  I_{2.5}(\rho,j) =
  \begin{cases}   
    \frac12 \Phi(\rho,j,F(\rho))    & \text{if }\div j = 0\\ +\infty
    & \text{otherwise}
\end{cases} , 
\end{equation}
with $\Phi$ again given by~\eqref{eqn:Phi_function}, together with~\eqref{eqn:dual_pairing},~\eqref{eqn:psi_star}
and~\eqref{eqn:new_psi}.

We have stated this LDP for joint fluctuations of the density and the current. For Markov chains, the LDP for the density and the
\emph{flow} is also known as a level-2.5 LDP~\cite{Bertini2015b}, so our general use of the name level-2.5
for~\eqref{equ:l2.5-general} may be non-standard, but it seems reasonable. The rate functional for the density and the current in
\eqref{equ:l2.5-general} can be obtained by contraction from the rate functional for the density and the flow (see Theorem 6.1
in~\cite{Bertini2015c}).

Using the splitting obtained in Section~\ref{sec:Decomp-rate-funct}, we obtain the following representation for the rate
functional on level-2.5.
\begin{proposition}
  \label{prop:split0_2.5}
  Let $j$ be divergence free. Then the level-2.5 rate functional~\eqref{eqn:ldp_2.5} is given by
  \begin{equation}
    \label{eqn:split_2.5}
    I_{2.5}(\rho,j) =\frac12\Bigl[\Phi_S\bigl(\rho,0,F^S(\rho)\bigr) + \Phi\bigl(\rho,j,F^A\bigr)\Bigr].
  \end{equation}
\end{proposition}

\begin{proof}
  We note from~\eqref{eq:div-free-D-null} that $D(\rho,j)$ vanishes for divergence free currents $j$. The result then directly
  follows from Corollary~\ref{cor:two}.  \qed
\end{proof}

\subsection{Large deviations for currents}
\label{sec:Large-devi-curr}

Proposition~\ref{prop:split0_2.5} is connected to recently-derived bounds on rate functions for currents,
{see~\cite{Gingrich2016a,Gingrich2017a,Pietzonka2016,Polettini2016a}}.  Indeed, the rate function for current
fluctuations can be obtained by contraction from level-2.5, as
\begin{equation}
  \label{equ:Ijj}
  {I_{\rm current}(j)}:=\inf_\rho I_{2.5}(\rho,j).
\end{equation}  
Then, following~\cite{Gingrich2017a,Polettini2016a}, it may be shown that for any $\rho,j,f$ one has for $\Phi$ as
in~\eqref{eqn:Phi_function} with~\eqref{eqn:dual_pairing},~\eqref{eqn:psi_star}-\eqref{eqn:new_psi} that
\begin{equation}
  \label{eqn:gingirch_bound}
  \Phi\bigl(\rho,j,f\bigr) \le \sum_{xy} \bigl(j_{xy}-j^{f}_{xy}(\rho)\bigr)^2 b_{xy}(\rho,f)
\end{equation}
with $b_{xy}(\rho,f)=f_{xy}/(4j^f_{xy}(\rho))$ if $f_{xy}\neq0$; otherwise $b_{xy}$ is continuously extended by taking
$b_{xy}(\rho,f)=1/(2a_{xy}(\rho))$.  Hence one has the result of~\cite{Gingrich2016a}, that the curvature of the rate function is
controlled by the housekeeping heat $F^A$, as
\begin{equation}
  \label{eqn:Ijj-bound}
  {I_{\rm current}(j)}\leq I_{2.5}(\pi,j) = \frac 1 2 \Phi\bigl(\pi,j,F^A\bigr) \leq 
  \frac 1 2 \sum_{xy} \frac{\bigl(j_{xy}-J_{xy}^{\rm ss}\bigr)^2}{4(J_{xy}^{\rm ss})^2} J_{xy}^{\rm ss} F^A_{xy},
\end{equation}
where $J^{\rm ss}:=J(\pi)$ is the steady state current (recall~\eqref{equ:master}), and the ratio $F^A_{xy}/J^{\ss}_{xy}$ must
again be interpreted as $2/a_{xy}(\rho)$ in the case where $F^A_{xy}$ (and hence $J^{\ss}_{xy}$) vanish.  The first step
in~\eqref{eqn:Ijj-bound} comes from~\eqref{equ:Ijj}, the second step uses~\eqref{eqn:split_2.5} as well as $\Phi(\pi,0,F^S)=0$,
and the third uses~\eqref{eqn:gingirch_bound}.

The significance of the splitting~\eqref{eqn:split_2.5} for this result is that $J^{\rm ss}_{xy}F_{xy}^A$ is the rate of flow of
housekeeping heat associated with edge $xy$: the appearance of the housekeeping heat is natural since the bound comes from the
second term in~\eqref{eqn:split_2.5}, which is independent of $F^S$ and depends only on $F^A$.

\subsection{Optimal control theory}
\label{sec:Optim-contr-theory}

It will be useful to introduce ideas of optimal control theory, whose relationship with large deviation theory is discussed
in~\cite{Fleming2006a,Chernyak2014a,Chetrite2015b,Jack2015a}.  In parallel with our given transition rates $r_{xy}$ we introduce a
new process, the \emph{controlled process}, where the rates are modified by a \emph{control potential} $\vhi$, as
\begin{equation}
  \label{equ:tilde-r}
  \tilde{r}_{xy} := r_{xy} \exp((\varphi(y)-\varphi(x))/2).
\end{equation}

For a given probability distribution $\rho$, we seek a potential $\vhi$ such that the controlled process has invariant measure
$\tilde\pi := \rho$. For this we need
\begin{equation*}
  \sum_y \left[\rho_x r_{xy} \exp((\varphi(y)-\varphi(x))/2) - \rho_y r_{yx} \exp((\varphi(x)-\varphi(y))/2)\right] = 0,
\end{equation*}
or equivalently
\begin{equation}
  \label{equ:vhi-balance}
  \div j^{F+\nabla \varphi}(\rho) = 
  \sum_y a_{xy}(\rho) \sinh\left( (F_{xy}(\rho) + \nabla^{x,y}\vhi)/2\right) = 0.
\end{equation}
We stress that, for any fixed $\rho$,~\eqref{equ:vhi-balance} is equivalent to solving the minimisation problem
\begin{equation}
  \label{eqn:minimisation_1}
  \inf_{\operatorname{div} j =0}\Phi\bigl(\rho,j,F(\rho)\bigr),
\end{equation}
which is also equivalent to maximisation of the Donsker-Varadhan functional, see for example Chapter IV.4
in~\cite{Hollander2000a}.  A proof for the existence and uniqueness of $\vhi$ can, e.g., be found in~\cite{Maes2012a}. Now assume
that $\vhi$ solves~\eqref{equ:vhi-balance}. The resulting controlled process depends on $\rho$ and has rates $\tilde{r}$ given
by~\eqref{equ:tilde-r}.  Throughout this section, we use tildes to indicate properties of the controlled process: all these
quantities depend implicitly on the fixed probability $\rho$.  Hence the (time-dependent) measure of the controlled process is
$\tilde\rho$.

Repeating the analysis of Section~\ref{sec:Setting} and noting that $\tilde{r}_{xy}\tilde{r}_{yx}=r_{xy}r_{yx}$, we find that
$\tilde a_{xy}(\tilde\rho) := 2\sqrt{\tilde\rho(x)\tilde r_{xy}\tilde\rho(y)\tilde r_{yx}} = a_{xy}(\tilde \rho)$.  Also, the
force for the controlled process is
\begin{equation}
  \label{equ:ftilde}
  \tilde{F}(\tilde\rho) = F(\tilde\rho) + \nabla \vhi,
\end{equation}
which may be decomposed as
\begin{equation}
  \begin{split}
    \tilde{F}^S(\tilde\rho) \;\!&:= F^S(\tilde\rho) + \nabla \log\frac\rho\pi =-\nabla \log \frac {\tilde\rho}{\rho},\\
    \tilde{F}^A \;\!& := F(\rho)+\nabla \varphi= F^A - \nabla \log\frac\rho\pi + \nabla\vhi.
  \end{split}
\end{equation}

Thus, the symmetric force in the controlled process vanishes when $\tilde\rho= \rho$. The antisymmetric force $\tilde{F}^A$
represents the force observed in the new non-equilibrium steady state $\rho$. If the original process is reversible, then
$\vhi=\log\frac\rho\pi$ so $\tilde{F}^A=F^A=0$.

It is useful to define $\tilde J_{xy}(\tilde\rho) := a_{xy}(\tilde\rho) \sinh (\tilde{F}_{xy}(\tilde\rho)/2)$ and to identify the
steady-state current for the controlled process as
\begin{equation}
  \label{equ:Jss}
  \tilde{J}^{\ss} := \tilde{J}(\rho) .
\end{equation}

\subsection{Decomposition of rate functions}
\label{sec:Decomp-rate-funct-1}

The ideas of optimal control theory are useful since they facilitate the further decomposition of the level-2.5 rate function into
several contributions.
\begin{lemma}
  \label{lemma:split1_2.5}
  Suppose that $\rho$ and $j$ are given and that $\div j=0$.  Then
  \begin{equation}
    \label{equ:i2.5-control}
    I_{2.5}(\rho,j) = 
    \frac12 \Bigl[ \Phi\bigl( \rho, \tilde{J}^{\ss} , F(\rho)\bigr) + \Phi\bigl( \rho, j , \tilde{F}^A \bigr) \Bigr],
  \end{equation}
  where $\tilde{J}^{\rm ss}$ is given by~\eqref{equ:Jss}, evaluated in the optimally controlled process whose steady state is
  $\rho$.
\end{lemma}

\begin{proof}
We write
\begin{align}
  2I_{2.5}(\rho,j)  & = \Psi(\rho,j) - j\cdot F(\rho) + \Psi^\star\bigl(\rho,F(\rho)\bigr) 
                      \nonumber \\
                    & = [\Psi(\rho,j) - j\cdot\tilde{F}(\rho) + \Psi^\star( \rho,\tilde{F}(\rho)) ] 
		      \nonumber \\
                      &\quad+ \Psi^\star(\rho,F(\rho)) - \Psi^\star( \rho,\tilde{F}(\rho)) - j\cdot(F(\rho)-\tilde{F}(\rho))
                      \nonumber \\
                    & = \Phi\bigl(\rho, j, \tilde{F}(\rho) \bigr)  
                      + \Psi^\star(\rho,F(\rho)) - \Psi^\star( \rho,\tilde{F}(\rho)) + j\cdot\nabla\varphi
                      \label{eqn:split-intermediate}
\end{align}
where the first line is~\eqref{eqn:Phi_function} and~\eqref{eqn:ldp_2.5}; the second line is simple rewriting; and the third uses
the definition of $\Phi$ in~\eqref{eqn:Phi_function} and also~\eqref{equ:ftilde} with $\tilde\rho=\rho$.

The current $\tilde{J}(\rho)$ satisfies $\Phi(\rho,\tilde{J}(\rho),\tilde{F}(\rho))=0$ so one has (by definition of $\Phi$) that
$\Psi^\star( \rho,\tilde{F}(\rho))=\tilde{J}(\rho)\cdot\tilde{F}(\rho)-\Psi(\rho,\tilde{J}(\rho))$.  Using this relation together
with~\eqref{equ:ftilde} and~\eqref{eqn:split-intermediate}, one has
\begin{equation}
  \qquad 2I_{2.5}(\rho,j)   = \Phi\bigl(\rho, j, \tilde{F}(\rho) \bigr)  
  + \Psi^\star(\rho,F(\rho)) - \tilde{J}(\rho)\cdot F(\rho)
  +\Psi(\rho,\tilde{J}(\rho)) -\tilde{J}(\rho)\cdot\nabla\varphi+ j\cdot\nabla\varphi.\qquad
\end{equation}
Finally we note that $\div\tilde{J}(\rho)=0$ (since $\rho$ is the invariant measure for the controlled process) and $\div j=0$ (by
assumption), so integration by parts yields $\tilde{J}(\rho)\cdot\nabla\varphi=0=j\cdot\nabla\varphi$; using once more the
definition of $\Phi$ yields~\eqref{eqn:final_markov_2}.  \qed
\end{proof}

The physical interpretation of~\eqref{equ:i2.5-control} is as follows. The contribution $\frac 12\Phi( \rho, j , \tilde{F}^A )$ is
a rate functional for observing an empirical current $j$ in the controlled process, while
$\frac 12\Phi( \rho, \tilde{J}^{\ss} , F(\rho) )$ is the rate functional for observing an empirical current $\tilde{J}^{\ss}$ in
the original process.  Since $\tilde{J}^{\ss}$ is the (deterministic) probability current for the controlled process, one has that
the more the controlled process differs from the original one, the larger will be $\Phi( \rho, \tilde{J}^{\ss}, F(\rho) )$.  Hence
the level-2.5 rate functional is large if the controlled process is very different from the original one, as one might expect. The
rate functional also takes larger values if the empirical current $j$ is very different from the probability current of the
controlled process.

We obtain our final representation for the level-2.5 rate functional, consisting of the sum of three different OM functionals.
\begin{proposition}
  \label{prop:final_markov_2.5} 
  Let $j$ be divergence free. We can represent the level-2.5 rate functional~\eqref{eqn:ldp_2.5} as
  \begin{equation}
    \label{eqn:final_markov_2.5}
    I_{2.5}(\rho,j) =  \frac12 \left[ \Phi_S\bigl( \rho, 0, F^S(\rho) \bigr) 
      +  \Phi\bigl(\rho,\tilde{J}^{\ss}, F^A \bigr) + \Phi\bigl( \rho, j , \tilde{F}^A \bigr) \right].
  \end{equation}
\end{proposition}

\begin{proof}
  This follows immediately from Lemma~\ref{lemma:split1_2.5} followed by an application of Corollary~\ref{cor:two} to
  $\Phi\bigl(\rho,\tilde{J}^{\ss}, F^A \bigr)$ and that $D=0$, from~\eqref{eq:div-free-D-null} .
\qed
\end{proof}

The three terms in~\eqref{eqn:final_markov_2.5} also appear in Lemma~\ref{lemma:split1_2.5} and Corollary~\ref{cor:two}, and their
interpretations have been discussed in the context of those results.  Briefly, we recall that $I_{2.5}(\rho,j)$ sets the
probability of fluctuations in which a non-typical density $\rho$ and current $j$ are sustained over a long time period.  The
first term in~\eqref{eqn:final_markov_2.5} reflects the fact that the free-energy gradient $F^S(\rho)$ tends to push $\rho$
towards the steady state $\pi$, so maintaining any non-typical density is unlikely if $F^S(\rho)$ is large.  Similarly, the second
term in~\eqref{eqn:final_markov_2.5} reflects the fact that large non-gradient forces $F^A$ also tend to suppress the probability
that $\rho$ maintains its non-typical value.  The final term is the only place in which the (divergence-free) current $j$ appears:
it vanishes if the current $j$ is typical within the controlled process (see Corollary~\ref{cor:final_markov_2}, below); otherwise
it reflects the probability cost of maintaining a non-typical circulating current.

\subsection{ Large deviations at level 2}
\label{sec:Large-deviations-at-level-2}

As well the LDP~\eqref{equ:l2.5-general}, we also consider an (apparently) simpler object, called a \emph{level-2 LDP}, where one
considers the density only. It is formally given by
\begin{equation}
  \label{equ:l2-general}
  \mathrm{Prob}\left(\hat\rho_T \approx \rho \right) \asymp \exp(-T I_{2}(\rho)).
\end{equation}
The contraction principle for LDPs~\cite[Section~3.6]{Touchette2009a} states that
\begin{equation}
  \label{equ:l2-inf}
  I_2(\rho) = \inf_{j \;\!:\;\! \div j=0}  I_{2.5}(\rho,j).
\end{equation}
Equation~\eqref{equ:i2.5-control} is uniquely minimised in its second argument for the divergence free current $j^{\tilde{F}^A}$,
such that the contraction over all divergence-free vector fields $j$ yields the level-2 rate functional
\begin{equation}
  \label{equ:i2-control}
  I_{2}(\rho) = \frac12 \Phi\bigl( \rho, \tilde{J}^{\ss} , F(\rho) \bigr).
\end{equation}
The same splitting as above finally allows us to write the level 2 rate functional as follows.
\begin{corollary}
  \label{cor:final_markov_2}
  The level-2 rate functional can be written as the sum
  \begin{equation}
    \label{eqn:final_markov_2}
    I_{2}(\rho) =  \frac12 \Bigl[\Phi_S\bigl( \rho, 0, F^S(\rho) \bigr) +  \Phi\bigl(\rho,\tilde{J}^{\ss}, F^A\bigr)\Bigr].
  \end{equation}
\end{corollary}

\begin{proof}
  This follows from~\eqref{equ:l2-inf} and~\eqref{eqn:final_markov_2.5}, since $\Phi\bigl( \rho, j , \tilde{F}^A \bigr)$ has a
  minimal value of zero.  \qed
\end{proof}

This last identity extends the results obtained in~\cite{Kaiser2017a} on the accelerated convergence to equilibrium for
irreversible processes using LDPs from the macroscopic scale (i.e.~in the regime of MFT) to Markov chains.  The level-2 rate
function in~\eqref{eqn:final_markov_2} can be interpreted as a rate of convergence to the steady state. It was shown
in~\cite{Kaiser2017a} that the rate is higher for irreversible processes, as opposed to reversible ones (as the second term
$\Phi(\rho,\tilde{J}^{\ss}, F^A)=0$ for reversible processes). We remark that splitting techniques for irreversible jump processes
have been used to devise efficient MCMC samplers; see for example~\cite{Bernard2011,Ma2016a}.

\subsection{Connection to MFT}
\label{sec:Level-2.5-level-2-in-MFT}

Under the assumption that no dynamical phase transition takes place, the time averaged density
{$\hat\rho_{[0,T]}^{L}:=\frac 1T\int_0^T\hat\rho_t^{L} \df t$} and current
{$\hat\jmath_{[0,T]}^{\lL}:=\frac 1T\int_0^T\hat\jmath_t^{\lL}\df t$} in MFT {(recall
  Section~\ref{sec:Large-devi-princ} for definitions) also satisfy a joint LDP in the limit $L,T\to\infty$: one takes first
  $L\to\infty$ and then $T\to\infty$, see \cite[Equ.~(36)]{Kaiser2017a}.  The LDP is similar to~\eqref{equ:l2.5-general}:
  
\begin{equation}
    \mathrm{Prob}\bigl((\hat\rho_{[0,T]}^L, \hat\jmath_{[0,T]}^L) \approx (\rho,j)\bigr) \asymp \exp\bigl\{-T |\Lambda_L| I_{\rm joint}^{\MFT}(\rho,j)\bigr\},
\end{equation}
where the rate function is, for a density profile $\rho$ and a current $j$ with $\operatorname{div}j=0$, given by
\begin{equation}
  \label{equ:lev2.5-mft}
     I_{\rm joint}^{\MFT}(\rho,j) = \frac12 \Phi_{\mathrm{MFT}}(\rho,j,F(\rho)).
\end{equation}
As for Markov chains (see Section~\ref{sec:Large-deviations-at}) $I_{\rm joint}^{\MFT}(\rho,j) =\infty$ if $j$ is not divergence
free. If $\div j=0$ then the rate function can be} written in the form~\cite{Kaiser2017a}
\begin{equation}
  \label{eqn:MFT_ldp_2.5}\quad
  { I_{\rm joint}^{\MFT}(\rho,j)} =\frac 14
  \int_\Lambda \nabla \frac{\delta \mathcal V}{\delta \rho}\cdot \chi\nabla \frac{\delta \mathcal V}{\delta \rho} \df x 
  + \frac 14\int_\Lambda  \nabla \varphi \cdot \chi\nabla \varphi \df x 
  + \frac14 \int_\Lambda (J_F-j)\cdot\chi^{-1} (J_F-j) \df x,\quad
\end{equation}
such that a contraction to {to the density only} yields
\begin{equation}
  \label{eqn:MFT_ldp_2}
  { I_{\rm density}^{\MFT}(\rho)} = 
  \frac 14\int_\Lambda \nabla \frac{\delta \mathcal V}{\delta \rho}\cdot \chi\nabla \frac{\delta \mathcal V}{\delta \rho} \df x 
  + \frac 14\int_\Lambda  \nabla \varphi \cdot \chi\nabla \varphi \df x.
\end{equation}
The function $\varphi$ in~\eqref{eqn:MFT_ldp_2.5} and~\eqref{eqn:MFT_ldp_2} is obtained by solving
\begin{equation}
  \div J_F(\rho) = 0, \qquad J_F(\rho) := \chi \nabla \varphi + J_A(\rho).
\end{equation}
Clearly the solution $\varphi$ depends on $\rho$.  In essence, we have reduced the minimisation problem~\eqref{equ:l2-inf} to the
solution of this PDE.  Comparing with~\eqref{eqn:final_markov_2.5}, we identify the terms $J_F = \chi\tilde{F}^A$ in the MFT
setting, and also $\tilde J^{\ss} = \chi \tilde{F}^A$, so $(\tilde{J}^{\ss}-\chi F^A(\rho)) = \chi\nabla\vhi$.  We obtain the
following representations for~\eqref{eqn:MFT_ldp_2.5} and~\eqref{eqn:MFT_ldp_2} reminiscent of
Proposition~\ref{prop:final_markov_2.5} and Corollary~\ref{cor:final_markov_2}.

\begin{proposition}
  The {rate functional for the joint density and current in MFT, which is given by~\eqref{eqn:MFT_ldp_2.5},} can be
  written in terms of the OM functional~\eqref{eqn:psi_mft} as
  \begin{equation}
    { I_{\rm joint}^{\MFT}(\rho,j)} =\frac 12\Bigl[\Phi_{\mathrm{MFT}}(\rho,0,F^S(\rho)) 
    + \Phi_{\mathrm{MFT}}(\rho,\tilde J^{\ss},F^A(\rho)) + \Phi_{\mathrm{MFT}}(\rho,j,\tilde F^A)\Bigr],
  \end{equation}
  and~\eqref{eqn:MFT_ldp_2}, the {rate functional for the density in MFT}, is given by
  \begin{equation}
    \label{equ:i2-mft-control}
    { I_{\rm density}^{\MFT}(\rho)} =\frac 12\Bigl[\Phi_{\mathrm{MFT}}(\rho,0,F^S(\rho)) + \Phi_{\mathrm{MFT}}(\rho,\tilde
      J^{\ss},F^A(\rho))\Bigr].
  \end{equation}
\end{proposition}

This proposition is equivalent to Proposition 5 of~\cite{Kaiser2017a}, but has now been rewritten in the language of optimal
control theory.  As discussed in~\cite{Kaiser2017a}, Equation~\eqref{equ:i2-mft-control} quantifies the extent to which breaking
detailed balance accelerates convergence of systems to equilibrium, at the hydrodynamic level. For this work, the key point is
that this result originates from Corollary~\ref{cor:final_markov_2}, which is the equivalent statement for Markov chains (without
taking any hydrodynamic limit).

\section{Consequences of the structure of the OM functional $\Phi$}
\label{sec:Cons-struct-OM}

We have shown that the rate functions for several LDPs in several different contexts depend on functionals $\Phi$ with the general
structure presented in~\eqref{eqn:Phi_function} and~\eqref{equ:legend}.  In this section, we show how this structure alone is
sufficient to establish some features that are well-known in MFT. This means that these results within MFT have analogues for
Markov chains. Our derivations mostly follow the standard MFT routes~\cite{Bertini2015a}, but we use a more abstract notation to
emphasise the minimal assumptions that are required.

\subsection{Assumptions}
\label{sec:assump-min}

The following minimal assumptions are easily verified for Markov chains; they are also either assumed or easily proven for MFT.
The results of this section are therefore valid in both settings.

We consider a process described by a time-dependent density $\rho$ and current $j$, with an associated continuity equation
$\dot\rho = -\div j$ and unique steady state $\pi$.  We are given a set of ($\rho$-dependent) forces denoted by $F(\rho)$, a dual
pairing $j\cdot f$ between forces and currents, and a function $\Psi(\rho,j)$ which is convex in $j$ and satisfies
$\Psi(\rho,j)=\Psi(\rho,-j)$.  With these choices, the functions $\Psi^\star$ and $\Phi$ are fully specified
via~\eqref{eqn:Phi_function} and~\eqref{equ:legend}.  We assume that for initial conditions chosen from the invariant measure, the
system satisfies an LDP of the form~\eqref{equ:pathwise-general} with rate function of the form~\eqref{eqn:mc_rate_functional}.

We define an adjoint process for which the probability of a path $(\rho_t,j_t)_{t\in[0,T]}$ is equal to the probability of the
time-reversed path $(\rho^*_t,j^*_t)_{t\in[0,T]}$ in the original process.  As above, we define
$(\rho^*_t,j^*_t)=(\rho_{T-t},-j_{T-t})$.  We assume that the adjoint process also satisfies an LDP of the
form~\eqref{equ:pathwise-general}, with rate function $I^*_{[0,T]}$.  Hence we must have
\begin{equation} 
  \label{equ:assume-adjoint}
  I^*_{[0,T]}\bigl((\rho_t, j_t)_{t\in[0,T]}\bigr) = I_{[0,T]}\bigl((\rho^*_t, j^*_t)_{t\in[0,T]}\bigr) .
\end{equation}
Moreover, we assume that $I^*_{[0,T]}$ may be obtained from $I$ by replacing the force $F(\rho)$ with some adjoint force
$F^*(\rho)$.  That is,
\begin{equation} 
  \label{equ:assume-Fstar}
  I^*_{[0,T]}\bigl((\rho_t, j_t)_{t\in[0,T]}\bigr)=  I_0(\rho_0) + \frac12\int_0^T \Phi(\rho_t,j_t, F^*(\rho_t))  \df t.
\end{equation}
Here, $I_0$ is the rate function associated with fluctuations of the density $\rho$, for a system in its steady state.  That is,
within the steady state, $\mathrm{Prob}(\hat\rho^{\;\!\cN}\approx \rho) \asymp \exp(-\cN I_0(\rho))$.  For Markov chains,
$I_0=\cal F$, the free energy; for MFT we have $I_0=\cal V$, the quasipotential. In the following we refer to $I_0$ as the free
energy.

\subsection{Symmetric and anti-symmetric forces}

Define 
\begin{equation} 
  \label{equ:Fs-Fa-gen}
  F^S(\rho) := \frac12[ F(\rho) + F^*(\rho) ], \qquad F^A(\rho) := \frac12[ F(\rho) - F^*(\rho) ].
\end{equation}
As the following proposition shows, $F^S$ is connected to the gradient of the free energy (or quasipotential) $I_0$, and the
forces $F^A$ and $F^S$ satisfy a generalised orthogonality (in the sense of Proposition~\ref{prop:orth}.)  The proof follows
Section II.C of~\cite{Bertini2015a}, but uses only the assumptions of Section~\ref{sec:assump-min}, showing that the result
applies also to Markov chains.

\begin{proposition}
  \label{prop:free_eng_balance}
  The forces $F^S$ and $F^A$ satisfy
  \begin{equation}
    \label{equ:Fsgrad}
    F^S(\rho) = -\nabla \frac{\delta I_0}{\delta \rho},
  \end{equation}
  and
    \begin{equation}
    \label{equ:FsFaOrth}
    \Psi^\star\bigl(\rho, F^S(\rho) + F^A\bigr) = \Psi^\star\bigl(\rho, F^S(\rho) - F^A\bigr) .
  \end{equation}%
\end{proposition}

\begin{proof}
Combining~\eqref{equ:assume-adjoint} and~\eqref{equ:assume-Fstar}, we obtain (for any path $(\rho_t,j_t)_{t\in[0,T]}$ that obeys
the continuity equation $\dot\rho = -\div j$)
\begin{equation}
  \label{equ:pathrev}
  I_0(\rho_0) + \frac12 \int_0^T\Phi(\rho_t,j_t,F(\rho_t))\df t 
  =I_0(\rho_T) + \frac12\int_0^T \Phi(\rho_{T-t},-j_{T-t},F^\ast(\rho_{T-t}))\df t.
\end{equation}
Differentiating with respect to $T$ and using~\eqref{eqn:Phi_function} together with $\Psi(\rho,j)=\Psi(\rho,-j)$
and~\eqref{equ:Fs-Fa-gen}, one has
\begin{equation*}
  \dot I_0(\rho) + j\cdot F^S(\rho) + \frac12 \big[ \Psi^\star(\rho,F^*(\rho)) - \Psi^\star(\rho,F(\rho)) \big] = 0 .
\end{equation*}
Using the continuity equation and an integration by parts, one finds $\dot I_0(\rho) =j\cdot \nabla \frac{\delta I_0}{\delta
  \rho}$, so that
\begin{equation*}
  j\cdot \left[ F^S(\rho) + \nabla \frac{\delta I_0}{\delta \rho} \right]
  + \frac12 \big[ \Psi^\star(\rho,F^*(\rho)) - \Psi^\star(\rho,F(\rho)) \big] = 0 .
\end{equation*}
This equation must hold for all $(\rho,j)$, which means that the two terms in square parentheses both vanish separately.
Combining the last equation with~\eqref{equ:Fs-Fa-gen}, we obtain~\eqref{equ:Fsgrad} and~\eqref{equ:FsFaOrth}.
\qed
\end{proof}

Proposition~\ref{prop:free_eng_balance} also yields a variational characterisation of $I_0$. The following corollary is analogous
to Equation~(4.8) of~\cite{Bertini2015a}, as is its proof.
\begin{corollary}
  \label{cor:var_free}
  The free energy $I_0$ satisfies 
{
  \begin{equation}
    \label{eqn:free_energy_var}
    I_0(\hat\rho) = \inf {\frac12\int_{-\infty}^0 \Phi(\rho_t,j_t, F(\rho_t))  \df t },
  \end{equation}
} 
where the infimum is taken over all paths $(\rho_t, j_t)_{t\in(-\infty,0]}$ that satisfy $\dot\rho_t+\div j_t=0$, as well as
$\lim_{t\to-\infty}\rho_t = \pi$ and $\rho_0 = \hat\rho$. Moreover, the optimal path is given by the time reversal of the solution
of the adjoint dynamics $(\rho_t, -J^\ast(\rho_t))_{t\in(-\infty,0]}$.
\end{corollary}

\begin{proof}
  We obtain from~\eqref{equ:pathrev} (together with~\eqref{eqn:mc_rate_functional} and~\eqref{equ:assume-adjoint}) that
{
  \begin{equation*}
    \frac12\int_{-\infty}^0 \Phi(\rho_t,j_t, F(\rho_t))  \df t
        = I_0(\hat\rho)  +  \frac12\int_{-\infty}^0 \Phi(\rho_t,j_t, F^*(\rho_t))  \df t.
  \end{equation*}
}

Taking the infimum on both sides yields~\eqref{eqn:free_energy_var}; indeed the infimum of
{$\frac12\int_{-\infty}^0 \Phi(\rho_t,j_t, F(\rho_t)) \df t$} is $0$, and this infimum is attained uniquely for the
optimal path for~\eqref{eqn:free_energy_var}. To see this, we note that $\Phi(\rho_t,-j_t, F^\ast(\rho_t))$ is uniquely minimised
for $j_t = -J^\ast(\rho_t)$, and $(\rho_t, -J^\ast(\rho_t))_{t\in(-\infty,0]}$ satisfies the conditions above, so the optimal path
is indeed the time-reversal of the solution of the adjoint dynamics.  \qed\end{proof}

\subsection{Hamilton-Jacobi like equation for the extended Hamiltonian}
\label{sec:Hamilt-Jacobi-equat}

Another important relationship within MFT is the Hamilton-Jacobi equation~\cite[Equation~(4.13)]{Bertini2015a}.  This provides a
characterisation of the quasipotential, as its maximal non-negative solution. The following formulation of that result uses only
the assumptions of Section~\ref{sec:assump-min} and therefore applies also to Markov chains. The functional
\begin{equation}
  \mathbb L(\rho,j):=\frac 12 \Phi(\rho,j,F(\rho))
  \label{equ:lag}
\end{equation}
can be interpreted as an extended Lagrangian (Note that $\mathbb L(\rho,j)$ should not be interpreted as a Lagrangian in the classical sense, as it depends on density and current $(\rho,j)$, rather than the pair consisting of density and associated velocity $(\rho,\dot\rho)$). {We follow Section~IV.G of~\cite{Bertini2015a}: given a sample path
  $(\rho_t,j_t)_{t\in[0,T]}$), define a vector field $A_t=A_0 - \int_0^t j_s \mathrm{d}s$.  The initial condition $A_0$ is chosen
  so that there is a bijection between the paths $(\rho_t,j_t)_{t\in[0,T]}$ and $(A_t)_{t\in[0,T]}$.  For example, in finite
  Markov chains, define $\bar\rho$ as a constant density, normalised to unity, and let $A_0=\nabla h$, where $h$ solves
  $\div(\nabla h) = (\rho_0-\bar\rho)$, see~\cite{Carlo2017a} for the relevant properties of these vector fields.  With this
  choice, and using $\dot\rho = -\div j$, one has $\rho_t=\bar\rho + \div A_t$ for all $t$, and one may also write (formally)
  $A_t = \div^{-1}(\rho_t-\bar\rho)$.  Comparing with~\cite[Section~IV.G]{Bertini2015a}, we write $\rho=\bar\rho+\div A$ instead
  of $\rho=\div A$ since for Markov chains one has (for any discrete vector field $A$) that $\sum_x \div A(x)=0$, so it is not
  possible to solve $\div A = \rho$ if $\rho$ is normalised to unity (recall that discrete vector fields have by definition
  $A_{xy}=-A_{yx}$~\cite{Carlo2017a}).
  
 The fluctuations
  of $A$ are therefore determined by the fluctuations of $(\rho,j)$, so the LDP \eqref{equ:pathwise-general} implies a similar LDP for $A$, whose  
  rate function is
  $I^{\rm ex}_{[0,T]}((A_t)_{t\in[0,T]}) = I^{\rm ex}_0(A_0) + \int_0^T \mathbb{L}^{\rm ex}(A_t,\dot A_t)\mathrm{d}t$, where $\mathbb{L}^{\rm ex}$ is {a Lagrangian that depends on $A$ and its time derivative (which we again refer to as extended Lagrangian, cf.~\cite{Bertini2015a}).}
  The function $\mathbb{L}$ in (\ref{equ:lag}) is then related to $\mathbb{L}^{\rm ex}$ via the bijection between $(\rho,j)$ and $A$.  Considering again the case of Markov chains, the time
  evolution of the system depends only on $\div A$ (which is $\rho-\bar\rho$)  and not on $A$ itself, one sees that $\mathbb{L}^{\rm ex}(A,\dot A)$
  depends only on $\div A$ and $\dot A$ (which is $j$).  Hence we write, formally,
  $\mathbb L(\rho,j) = \mathbb{L}^{\rm ex}(\div^{-1}(\rho-\bar\rho),-j)$, and we recover~\eqref{equ:lag}.  

Hence $\mathbb{L}$ is nothing but the extended Lagrangian $\mathbb{L}^{\rm ex}$, written in different variables: for this reason we {refer to $\mathbb{L}$ as an (extended) Lagrangian.}

  To arrive at the corresponding {(extended) Hamiltonian}, one should write
  $\mathbb{H}^{\rm ex}(A,\xi) = \sup_{\dot A} [ \xi \cdot \dot{A} - \mathbb{L}^{\rm ex}(A_t,\dot A_t) ]$, or equivalently }
\begin{equation}
  \mathbb H(\rho,\xi) =\sup_j \bigl( j \cdot \xi - \mathbb L(\rho,j)\bigr),
\end{equation}
where $\xi$ is a conjugate field for the current $j$. We identify $\mathbb{H}$ as the scaled cumulant generating function
associated with the rate function $I_{2.5}(\rho,j)=\mathbb{L}(\rho,j)$~\cite[Section 3.1]{Touchette2009a}.  Analysis of rare
fluctuations in terms of the field $\xi$ is often more convenient than direct analysis of the rate
function~\cite{Lebowitz1999a,Lecomte2007a} and is the basis of the ``$s$-ensemble'' method that has recently been exploited in a
number of physical applications (for example~\cite{Garrahan2009a,Jack2015a}). Using~\eqref{eqn:Phi_function}
and~\eqref{equ:legend}, we obtain
\begin{equation}
  \label{equ:H-psi*}
  \mathbb H(\rho,\xi) =\frac 12\Psi^\star(\rho,F(\rho)+2\xi) - \frac 12\Psi^\star(\rho,F(\rho)).
\end{equation}
(This generalises the definition~\eqref{equ:ham-markov}, which was restricted to Markov chains.)

To relate this {extended Hamiltonian} to the free energy (quasipotential), {one can define an \emph{extended Hamilton-Jacobi equation}, which is for a functional
$\cal S$ given by}
\begin{equation}
  \label{eqn:HJ_micro}
  \mathbb H\left(\rho,\nabla\frac{\delta \mathcal S}{\delta\rho}\right)=0.
\end{equation}
The relation of this equation to the free energy is given by the following proposition, which mirrors Equation~(4.18)
of~\cite{Bertini2015a}, but now in our generalised setting, so that it applies also to Markov chains.

\begin{proposition}
  \label{prop:HJ}
  The free energy $I_0$ is the maximal non-negative solution to~\eqref{eqn:HJ_micro} which vanishes at the steady state $\pi$. In
  other words, any functional $\mathcal S$ that solves~\eqref{eqn:HJ_micro} and has $\mathcal S(\pi)=0$ also satisfies
  $\mathcal S\le I_0$.
\end{proposition}

\begin{proof}
  From~\eqref{equ:Fs-Fa-gen},~\eqref{equ:Fsgrad},~\eqref{equ:FsFaOrth} and $\Psi^\star(\rho,F)=\Psi^\star(\rho,-F)$, one has
\begin{equation}
  \label{equ:HJ-solve}
  \Psi^\star(\rho,F(\rho)+2\nabla\tfrac{\delta I_0}{\delta\rho})=\Psi^\star(\rho,-F_S(\rho)+F_A(\rho))=\Psi^\star(\rho,F(\rho)).
\end{equation} 
Thus~\eqref{equ:H-psi*} yields $\mathbb H\bigl(\rho,\nabla\tfrac{\delta I_0}{\delta\rho}\bigr)=0$, so $I_0$ does indeed
solve~\eqref{eqn:HJ_micro}. In addition,~\eqref{equ:HJ-solve} is valid also with $I_0$ replaced by any $\cal S$ that
solves~\eqref{eqn:HJ_micro}; combining this result with~\eqref{eqn:Phi_function} yields
\begin{equation}
  \label{equ:Phi-HJ-bound}
  \Phi(\rho,j,F(\rho)) = \Phi\left(\rho,j,F(\rho) +2\nabla \frac{\delta\mathcal S}{\delta\rho}\right)
  + 2 j\cdot \nabla \frac{\delta\mathcal S}{\delta\rho}\ge 2 j\cdot \nabla \frac{\delta\mathcal S}{\delta\rho},
\end{equation}
where the second step uses $\Phi\geq0$.  Moreover, for any path $(\rho_t,j_t)_{t\in(-\infty,0]}$ with $\dot\rho_t+\div j_t=0$ and
  $\lim_{t\to-\infty}\rho_t = \pi$, we have from~\eqref{equ:Phi-HJ-bound} that
\begin{equation}\nonumber
\qquad  I_{(-\infty,0]}\bigl((\rho,j)_{t\in(-\infty,0]}\bigr) 
      = \int_{-\infty}^0 \Phi(\rho_t,j_t,F(\rho_t)) \df t  \ge
      \int_{-\infty}^0  j(x)\cdot \nabla \frac{\delta\mathcal S}{\delta\rho}(x) \df t  = \mathcal S(\rho_0),\qquad
\end{equation}
where the final equality uses an integration by parts, together with the continuity equation.  Finally, taking the infimum over
all paths and using Corollary~\ref{cor:var_free}, one obtains $\mathcal{S}(\rho) \leq I_0(\rho)$, as claimed.  \qed\end{proof}

{
\subsection{Generalisation of Lemma~\ref{lem:psi_split}}

Before ending, we note that~\eqref{equ:FsFaOrth} is analogous to Proposition~\ref{prop:orth} in the general setting of this
section, but we have not yet proved any analogue of Lemma~\ref{lem:psi_split}. Hence we have not obtained a generalisation of
Corollary~\ref{cor:two}, nor any of its further consequences.  To achieve this, one requires a further assumption within the
general framework considered here, which amounts to a splitting of the Hamiltonian.  This assumption holds for MFT and for Markov
chains, and is a sufficient condition for a generalised Lemma~\ref{lem:psi_split}. 

To state the assumption, we consider a reversible process in which the forces are $F^S(\rho)$.  (For Markov chains we should
consider the process with rates $r^S_{xy} = \frac12( r_{xy} + r_{xy}^*)$; for MFT it is the process with $J(\rho)=J^S(\rho)$ and
the same mobility $\chi$ as the original process.) We assume that such a process exists and that its Hamiltonian can be written as
$\mathbb{H}_S(\rho,\xi) = \frac12 [ \Psi^\star_S(\rho,F^S(\rho) + 2\xi) - \Psi_S^\star(\rho,F^S(\rho))]$ for some function
$\Psi^\star_S$ (compare~\eqref{equ:H-psi*} and see Section~\ref{sec:HJ-for-MC} for the case of Markov chains).  Also let the
Hamiltonian for the adjoint process be $\mathbb{H}^*(\rho,\xi)$, which is constructed by replacing $F$ by $F^*$
in~\eqref{equ:H-psi*}.  Then, one assumes further that
\begin{equation}
  \mathbb{H}_S(\rho,\xi) = \tfrac12 [ \mathbb{H}(\rho,\xi) + \mathbb{H}^*(\rho,\xi) ] ,
\end{equation}
which may be verified to hold for Markov chains and for MFT.  Writing $\xi=-F^S/2$ and using~\eqref{equ:H-psi*}
with~\eqref{equ:FsFaOrth} and $\Psi^\star(\rho,f) = \Psi^\star(\rho,-f)$, one then obtains
\begin{equation}
  \Psi^\star_S(\rho,F^S(\rho)) = \Psi^\star(F(\rho)) - \Psi^\star(F^A(\rho)) , 
\end{equation}
which is the promised generalisation of Lemma~\ref{lem:psi_split}.
}

\section{Conclusion}
\label{sec:conc}

In this article, we have presented several results for dynamical fluctuations in Markov chains.  The central object in our
discussion has been the function $\Phi$, which plays a number of different roles -- it is the rate function for large deviations
at level 2.5 (Equation~\eqref{eqn:ldp_2.5}), and it also appears in the rate function for pathwise large deviation functions
(Equation~\eqref{eqn:mc_rate_functional}).  These results -- derived originally by Maes and co-workers~\cite{Maes2008a,Maes2008b}
-- originate from the relationship between $\Phi$ and the relative entropy between path measures (Appendix~\ref{sec:relent}).  The
canonical (Legendre transform) structure of $\Phi$ (Equation~\eqref{equ:legend}) and its relation to time reversal
(Equation~\eqref{equ:gc-finite-time}) have also been discussed before~\cite{Maes2008a}.

The function $\Phi$ depends on probability currents $j$ and their conjugate forces $f$.  Our Proposition~\ref{prop:orth} and
Corollary~\ref{cor:two} show how the rate functions in which $\Phi$ appears have another level of structure, based on the
decomposition of the forces $F$ in two pieces $F=F^S+F^A$, according to its behaviour under time-reversal.  A similar
decomposition is applied in Macroscopic Fluctuation Theory~\cite{Bertini2015a}: the discussion of
Sections~\ref{sec:LDPs-time-averaged} and~\ref{sec:Cons-struct-OM} show how several results of that theory -- which applies on
macroscopic (hydrodynamic) scales -- already have analogues for Markov chains, which provide microscopic descriptions of
interacting particle systems.  These results -- which concern symmetries, gradient structures and (generalised) orthogonality
relationships -- show how properties of the rate functions are directly connected to physical ideas of free energy, dissipation,
and time-reversal.

Looking forward, we hope that these structures can be exploited both in mathematics and physics.  From a mathematical viewpoint,
the canonical structure and generalised orthogonality relationships may provide new routes for scale-bridging calculations, just
as the geometrical structure identified by Maas~\cite{Maas2011a} has been used to develop new proofs of hydrodynamic
limits~\cite{Fathi2016a}.  In physics, a common technique is to propose macroscopic descriptions of physical systems based on
symmetries and general principles -- examples in non-equilibrium (active) systems include~\cite{Toner1995a,Wittkowski2014a}.
However, this level of description leaves some ambiguity as to the best definitions of some physical quantities, such as the local
entropy production~\cite{Nardini2017a}.  We hope that the structures identified here can be useful in relating such macroscopic
theories to underlying microscopic behaviour.

\paragraph{Acknowledgements}
We thank Freddy Bouchet, Davide Gabrielli, Juan Garrahan, Jan Maas, Michiel Renger and Hugo Touchette
for useful discussions. MK is supported by a scholarship from the EPSRC Centre for Doctoral Training in Statistical Applied
Mathematics at Bath (SAMBa), under the project EP/L015684/1.  JZ gratefully acknowledges funding by the EPSRC through project
EP/K027743/1, the Leverhulme Trust (RPG-2013-261) and a Royal Society Wolfson Research Merit Award. {The authors thank the anonymous referees for their careful reading of the manuscript and for many helpful comments and suggestions.}

\appendix

\section{Relative entropy on path space}
\label{sec:relent}

Consider a Markov process with rates $r(x,y)$ and initial distribution $Q_0$. We fix a time interval $[0,T]$ for some $T>0$ and
denote the distribution of the Markov process on this time interval with $Q$. For each path $(x_u)_{u\in[0,T]}$ with jumps at
times $t_1,\dots, t_n$ the density of $Q$ can be found by solving the associated master equation~\eqref{eq:master}; it is given by
\begin{equation*}
  Q\bigl( (x_u)_{u\in[0,T]}\bigr)
  = Q_0(x_0)\exp\biggl\{\int_0^T \biggl(\sum_{i=1}^n\log r_t(x_{t-},x_t) \delta(t-t_i) - \sum_y r_t(x_t,y) \biggr)\df t\biggr\},
\end{equation*}
where $x_{t-}:=\lim_{\epsilon\to 0}x_{t-\epsilon}$ is the state of the process just before time $t$.

Now consider a second Markov process with time-dependent rates $\hat r_t(x,y)$ and initial distribution $P_0$.  The distribution
of this process is denoted by $P$. The logarithmic density of $P$ with respect to $Q$ is given by
\begin{multline}\nonumber
  \log \frac{dP}{dQ}\bigl((x_u)_{u\in[0,T]}\bigr) 
  =\log \frac{dP_0}{dQ_0}(x_0) \\
+ \int_{0}^{T} \biggl(\sum_{i=1}^n\log \Bigl(\frac{\hat r_t(x_{t^-},x_{t})}{r(x_{t^-},x_{t})}\Bigr)\;\! \delta(t-t_i)
  - \sum_y \bigl[\hat r_t(x_t,y)-r(x_t,y)\bigr]\biggr)\df t.
\end{multline}
We further denote the distribution of $P$ at time $t$ with $\rho_t$, such that $\rho_t = P\circ X_t^{-1}$ where $X_t$ denotes the
evaluation of the path at time $t$ (such that in particular $P_0=\rho_0$).  The \emph{relative entropy} on path space
\begin{equation*}
  \mathcal H(P|Q) := \mathbb E_P\biggl[\log \Bigl(\frac{dP}{dQ}\Bigr)\biggr]
\end{equation*}
is then equal to
\begin{equation*}
 \mathbb E_{P_0}\biggl[\log \Bigl(\frac{dP_0}{dQ_0}\Bigr)\biggr] 
  + \int_0^T \sum_{x,y}\rho_t(x)\Bigl(\hat r_t(x,y)\log\Bigl(\frac{\hat r_t(x,y)}{r(x,y)}\Bigr)-\hat r_t(x,y)+r(x,y)\Bigr)\df t.
\end{equation*}
Let $(\rho_t,j_t)_{t\in[0,T]}$ be given, with $\rho_t>0$ for all times $t\in[0,T]$. We then can rewrite the relative entropy $\mathcal H(P|Q)$ in
terms of the flow $C_t(x,y):=\rho_t(x)\hat r_t(x,y)$ as
\begin{equation}
  \label{eqn:rf_1}
  \mathcal H(\rho_0|Q_0) + 
  \int_0^T\sum_{x,y}\Bigl(C_t(x,y)\log \Bigl(\frac{C_t(x,y)}{\rho_t(x)r(x,y)}\Bigr) - C_t(x,y) + \rho_t(x)r(x,y) \Bigr)\df t.
\end{equation}
Note that the relative entropy $\mathcal H(P|Q)$ can (just as the Markov chain) be completely characterised by the probability
distribution $(\rho_t)_{t\in[0,T]}$ and the flow $(C_t)_{t\in[0,T]}$.

We are interested in a special flow $(C_t)_{t\in[0,T]}$ which recovers a given current $(j_t)_{t\in[0,T]}$ as
$(j_t)_{xy}=C_t(x,y)-C_t(y,x)$.  The force associated to $j_t$ is by~\eqref{eq:J-sinh} given by
$f^{j_t}(\rho_t):=2 \operatorname{arcsinh}(j_t/a(\rho_t))$ and the flow of interest is defined as
$C_t(x,y) = \frac12 a_{xy}(\rho_t) \exp(\frac 12 f_{xy}^{j_t}(\rho_t))$. It can be interpreted as the optimal flow that creates
the current $(j_t)_{t\in[0,T]}$.

We define the rates $\tilde r_t(x,y):=C_t(x,y)/\rho_t(x)$ and denote the law of the associated (time heterogeneous) Markov process
on $[0,T]$ with $\tilde P$. The relative entropy of this new process $\tilde P$ with respect to the reference process $Q$ is
\begin{equation}
  \label{eqn:rf_2}
  \mathcal H(\tilde P | Q) = \mathcal H(\rho_0|Q_0) + \frac 12 \int_0^T \Phi(\rho_t,j_t, F(\rho_t)) \df t
\end{equation}
with $\Phi$ given by~\eqref{eqn:Phi_function}; to see this, we argue as follows.  Symmetrising~\eqref{eqn:rf_1} and considering
each summand separately gives
\begin{equation}\nonumber
\qquad  \frac 12\Bigl( C_t(x,y)\log \frac{C_t(x,y)}{C_t^Q(x,y)} + C_t(y,x)\log \frac{C_t(y,x)}{C_t^Q(y,x)} \Bigr) 
  + \frac 12 \Bigl( C_t^Q(x,y) - C_t(x,y) + C_t^Q(y,x) - C_t(y,x) \Bigr),\qquad
\end{equation}
where the first summand coincides with
\begin{equation*}
  \frac 12 \biggl( \frac 12 a_{xy}(\rho_t)\sinh\bigl(\tfrac 12f^{j_t}_{xy}(\rho_t)\bigr) f^{j_t}_{xy}(\rho_t)
  - \frac 12 a_{xy}(\rho_t)\sinh\bigl(\tfrac 12 f^{j_t}_{xy}(\rho_t)\bigr) F_{xy}(\rho_t)\biggr)
\end{equation*}
and the second is given by
\begin{equation*}
  \frac 12 \Bigl(a_{xy}(\rho_t)\cosh\bigl(\tfrac 12F_{xy}(\rho_t)\bigr)
  -a_{xy}(\rho_t)\cosh\bigl(\tfrac 12f^{j_t}_{xy}(\rho_t)\bigr)\Bigr).
\end{equation*}
Combining this with~\eqref{equ:parts} and~\eqref{eqn:psi_star} yields~\eqref{eqn:rf_2}.

\noindent\emph{Pathwise Large Deviation Principle}: Let $x^1,x^2,\dots$ be a sequence of iid copies of the Markov chains with law
$Q$.  By Sanov's Theorem (see, e.g., Theorem~6.2.10 in~\cite{Dembo2010a}), the empirical average
$\frac 1{\mathcal N}\sum_{i=1}^\mathcal N \delta_{x^i}$ of the Markov chains satisfies a LDP with the rate functional
$\mathcal H(\cdot | Q)$. We can interpret $\mathcal H(\cdot | Q)$ as the rate functional for the joint LDP of
$(\rho_t,C_t)_{t\in[0,T]}$ by defining this rate functional $\mathcal I_{[0,T]}((\rho_t,C_t)_{t\in[0,T]})$ as the right-hand side
of~\eqref{eqn:rf_1}.

We contract the above rate functional to obtain the rate functional for the joint empirical measure and current
$(\rho_t,j_t)_{t\in[0,T]}$. It is given by
\begin{equation}
  I_{[0,T]}((\rho_t,j_t)_{t\in[0,T]}):=\inf_{(C_t)_{t\in[0,T]}} \mathcal I_{[0,T]}((\rho_t,C_t)_{t\in[0,T]}),
\end{equation}
where the infimum is taken over the set of all flows which yield the current $(j_t)_{t\in[0,T]}$, i.e. over the set
$\{(C_t)_{t\in[0,T]} |$ for all $t\in[0,T]: C_t(x,y)\ge 0$ and $ C_t(x,y) - C_t(y,x) = (j_t)_{xy}\}$.  It was shown
in~\cite{Maes2008a} and~\cite{Bertini2015b} that the minimising flow is the current
$C_t(x,y) = \frac12 a_{xy}(\rho_t) \exp(\frac 12 f_{xy}^{j_t}(\rho_t))$ introduced above, such that
$I_{[0,T]}((\rho_t,j_t)_{t\in[0,T]})$ coincides with~\eqref{eqn:rf_2}.

\small
\setstretch{1}
\bibliographystyle{plain}

\begin{thebibliography}{10}


\bibitem{Andrieux2007a}
David Andrieux and Pierre Gaspard.
\newblock Fluctuation theorem for currents and {S}chnakenberg network theory.
\newblock {\em J. Stat. Phys.}, 127(1):107--131, 2007.

\bibitem{Basile2017}
Giada Basile, Dario Benedetto, and Lorenzo Bertini.
\newblock A gradient flow approach to linear boltzmann equations.
\newblock {\em arXiv preprint arXiv:1707.09204}, 2017.

\bibitem{Basu2015a}
Urna Basu and Christian Maes.
\newblock Nonequilibrium response and frenesy.
\newblock {\em J. Phys. Conf. Ser.}, 638(1):012001, 2015.

\bibitem{Benois1995}
O~Benois, C~Kipnis, and C~Landim.
\newblock Large deviations from the hydrodynamical limit of mean zero
  asymmetric zero range processes.
\newblock {\em Stochastic processes and their applications}, 55(1):65--89,
  1995.

\bibitem{Bernard2011}
Etienne~P. Bernard and Werner Krauth.
\newblock Two-step melting in two dimensions: First-order liquid-hexatic
  transition.
\newblock {\em Phys. Rev. Lett.}, 107:155704, 2011.

\bibitem{Bertini2009}
Lorenzo Bertini, Claudio Landim, Mustapha Mourragui, et~al.
\newblock Dynamical large deviations for the boundary driven weakly asymmetric
  exclusion process.
\newblock {\em The Annals of Probability}, 37(6):2357--2403, 2009.

\bibitem{Bertini2015a}
Lorenzo Bertini, Alberto De~Sole, Davide Gabrielli, Giovanni Jona-Lasinio, and
  Claudio Landim.
\newblock Macroscopic fluctuation theory.
\newblock {\em Rev. Modern Phys.}, 87(2):593--636, 2015.

\bibitem{Bertini2015c}
Lorenzo Bertini, Alessandra Faggionato, and Davide Gabrielli.
\newblock Flows, currents, and cycles for {M}arkov chains: large deviation
  asymptotics.
\newblock {\em Stochastic Process. Appl.}, 125(7):2786--2819, 2015.

\bibitem{Bertini2015b}
Lorenzo Bertini, Alessandra Faggionato, and Davide Gabrielli.
\newblock Large deviations of the empirical flow for continuous time {M}arkov
  chains.
\newblock {\em Ann. Inst. Henri Poincar\'e Probab. Stat.}, 51(3):867--900,
  2015.

\bibitem{Chernyak2014a}
Vladimir~Y. Chernyak, Michael Chertkov, Joris Bierkens, and Hilbert~J. Kappen.
\newblock Stochastic optimal control as non-equilibrium statistical mechanics:
  calculus of variations over density and current.
\newblock {\em J. Phys. A}, 47(2):022001, 2014.

\bibitem{Chetrite2015b}
Rapha{\"e}l Chetrite and Hugo Touchette.
\newblock Variational and optimal control representations of conditioned and
  driven processes.
\newblock {\em J. Stat. Mech. Theory Exp.}, 2015(12):P12001, 42, 2015.

\bibitem{Crooks2000a}
Gavin~E. Crooks.
\newblock Path-ensemble averages in systems driven far from equilibrium.
\newblock {\em Phys. Rev. E}, 61:2361--2366, 2000.

\bibitem{Carlo2017a}
Leonardo De~Carlo and Davide Gabrielli.
\newblock Gibbsian {S}tationary {N}on-equilibrium {S}tates.
\newblock {\em J. Stat. Phys.}, 168(6):1191--1222, 2017.

\bibitem{Dembo2010a}
Amir Dembo and Ofer Zeitouni.
\newblock {\em Large deviations techniques and applications}, volume~38 of {\em
  Stochastic Modelling and Applied Probability}.
\newblock Springer-Verlag, Berlin, 2010.
\newblock Corrected reprint of the second (1998) edition.

\bibitem{Hollander2000a}
Frank den Hollander.
\newblock {\em Large deviations}, volume~14 of {\em Fields Institute
  Monographs}.
\newblock American Mathematical Society, Providence, RI, 2000.

\bibitem{Esposito2010a}
Massimiliano Esposito and Christian Van~den Broeck.
\newblock Three faces of the second law. {I}. {M}aster equation formulation.
\newblock {\em Phys. Rev. E}, 82:011143, 2010.

\bibitem{Fathi2016a}
Max Fathi and Marielle Simon.
\newblock {\em The Gradient Flow Approach to Hydrodynamic Limits for the Simple
  Exclusion Process}, pages 167--184.
\newblock Springer International Publishing, Cham, 2016.

\bibitem{Fleming2006a}
Wendell~H. Fleming and H.~Mete Soner.
\newblock {\em Controlled {M}arkov processes and viscosity solutions},
  volume~25 of {\em Stochastic Modelling and Applied Probability}.
\newblock Springer, New York, second edition, 2006.

\bibitem{Gallavotti1995a}
G.~Gallavotti and E.~G.~D. Cohen.
\newblock Dynamical ensembles in stationary states.
\newblock {\em J. Statist. Phys.}, 80(5-6):931--970, 1995.

\bibitem{Gardiner2009a}
Crispin Gardiner.
\newblock {\em Stochastic methods}.
\newblock Springer Series in Synergetics. Springer-Verlag, Berlin, fourth
  edition, 2009.
\newblock A handbook for the natural and social sciences.

\bibitem{Garrahan2009a}
Juan~P. Garrahan, Robert~L. Jack, Vivien Lecomte, Estelle Pitard, Kristina van
  Duijvendijk, and Fr{\'e}d{\'e}ric van Wijland.
\newblock First-order dynamical phase transition in models of glasses: an
  approach based on ensembles of histories.
\newblock {\em J. Phys. A}, 42(7):075007, 34, 2009.

\bibitem{Gingrich2016a}
Todd~R. Gingrich, Jordan~M. Horowitz, Nikolay Perunov, and Jeremy~L. England.
\newblock Dissipation bounds all steady-state current fluctuations.
\newblock {\em Phys. Rev. Lett.}, 116:120601, 2016.

\bibitem{Gingrich2017a}
Todd~R Gingrich, Grant~M Rotskoff, and Jordan~M Horowitz.
\newblock Inferring dissipation from current fluctuations.
\newblock {\em J. Phys. A}, 50(18):184004, 2017.

\bibitem{Jack2015a}
R.~L. Jack and P.~Sollich.
\newblock Effective interactions and large deviations in stochastic processes.
\newblock {\em Eur. Phys. J. Spec. Top.}, 224(12):2351--2367, 2015.

\bibitem{Jarzynski1997a}
C.~Jarzynski.
\newblock Nonequilibrium equality for free energy differences.
\newblock {\em Phys. Rev. Lett.}, 78:2690--2693, 1997.

\bibitem{Kaiser2017a}
Marcus Kaiser, Robert~L. Jack, and Johannes Zimmer.
\newblock Acceleration of convergence to equilibrium in markov chains by
  breaking detailed balance.
\newblock {\em J. Stat. Phys.}, 168:259--287, 2017.

\bibitem{Kipnis1989}
C~Kipnis, S~Olla, and SRS Varadhan.
\newblock Hydrodynamics and large deviation for simple exclusion processes.
\newblock {\em Communications on Pure and Applied Mathematics}, 42(2):115--137,
  1989.

\bibitem{Kipnis1999a}
Claude Kipnis and Claudio Landim.
\newblock {\em Scaling limits of interacting particle systems}, volume 320 of
  {\em Grundlehren der Mathematischen Wissenschaften [Fundamental Principles of
  Mathematical Sciences]}.
\newblock Springer-Verlag, Berlin, 1999.
\bibitem{fisher07}

A.~B.~ Kolomeisky and M.~E.~Fisher.  
\newblock Molecular Motors: A Theorist's Perspective.
\newblock \emph{Ann. Rev. Phys. Chem.} 58(1):675-695, 2007

\bibitem{Kwon2005a}
Chulan Kwon, Ping Ao, and David~J. Thouless.
\newblock Structure of stochastic dynamics near fixed points.
\newblock {\em PNAS}, 102(37):13029--13033, 2005.

\bibitem{lavorel76}
J. Lavorel.
\newblock Matrix analysis of the oxygen evolving system of photosynthesis.
\newblock J. Theor. Biol. 57(1):171-185, 1976

\bibitem{Lebowitz1999a}
Joel~L. Lebowitz and Herbert Spohn.
\newblock A {G}allavotti-{C}ohen-type symmetry in the large deviation
  functional for stochastic dynamics.
\newblock {\em J. Statist. Phys.}, 95(1-2):333--365, 1999.

\bibitem{Lecomte2007a}
V.~Lecomte, C.~Appert-Rolland, and F.~van Wijland.
\newblock Thermodynamic formalism for systems with {M}arkov dynamics.
\newblock {\em J. Stat. Phys.}, 127(1):51--106, 2007.

\bibitem{Ma2016a}
Yi-An Ma, Emily~B. Fox, Tianqi Chen, and Lei Wu.
\newblock A unifying framework for devising efficient and irreversible {MCMC}
  samplers, 2016.

\bibitem{Maas2011a}
Jan Maas.
\newblock Gradient flows of the entropy for finite {M}arkov chains.
\newblock {\em J. Funct. Anal.}, 261(8):2250--2292, 2011.

\bibitem{Machlup1953a}
S.~Machlup and L.~Onsager.
\newblock Fluctuations and irreversible process. {II}. {S}ystems with kinetic
  energy.
\newblock {\em Physical Rev. (2)}, 91:1512--1515, 1953.

\bibitem{Maes2008a}
C.~Maes and K.~Neto{\v{c}}n\'y.
\newblock Canonical structure of dynamical fluctuations in mesoscopic
  nonequilibrium steady states.
\newblock {\em Europhys. Lett. EPL}, 82(3):Art. 30003, 6, 2008.

\bibitem{Maes2008b}
C.~Maes, K.~Neto{\v{c}}n\'y, and B.~Wynants.
\newblock On and beyond entropy production: the case of {M}arkov jump
  processes.
\newblock {\em Markov Process. Related Fields}, 14(3):445--464, 2008.

\bibitem{Maes1999a}
Christian Maes.
\newblock The fluctuation theorem as a {G}ibbs property.
\newblock {\em J. Statist. Phys.}, 95(1-2):367--392, 1999.

\bibitem{Maes2012a}
Christian Maes, Karel Neto{\v{c}}n\'y, and Bram Wynants.
\newblock Monotonicity of the dynamical activity.
\newblock {\em J. Phys. A}, 45(45):455001, 13, 2012.

\bibitem{Maes2015a}
Christian Maes and Karel Neto{\v{c}}n\'y.
\newblock Revisiting the {G}lansdorff-{P}rigogine criterion for stability
  within irreversible thermodynamics.
\newblock {\em J. Stat. Phys.}, 159(6):1286--1299, 2015.

\bibitem{Mariani2012}
Mauro Mariani.
\newblock A gamma-convergence approach to large deviations.
\newblock {\em arXiv preprint arXiv:1204.0640}, 2012.

\bibitem{Mielke2014a}
A.~Mielke, M.~A. Peletier, and D.~R.~M. Renger.
\newblock On the relation between gradient flows and the large-deviation
  principle, with applications to {M}arkov chains and diffusion.
\newblock {\em Potential Anal.}, 41(4):1293--1327, 2014.

\bibitem{Nardini2017a}
Cesare Nardini, \'Etienne Fodor, Elsen Tjhung, Fr\'ed\'eric van Wijland, Julien
  Tailleur, and Michael~E. Cates.
\newblock Entropy production in field theories without time-reversal symmetry:
  Quantifying the non-equilibrium character of active matter.
\newblock {\em Phys. Rev. X}, 7:021007, 2017.

\bibitem{Pietzonka2016}
Patrick Pietzonka, Andre~C. Barato, and Udo Seifert.
\newblock Universal bounds on current fluctuations.
\newblock {\em Phys. Rev. E}, 93:052145,  2016.

\bibitem{Polettini2016a}
Matteo Polettini, Alexandre Lazarescu, and Massimiliano Esposito.
\newblock Tightening the uncertainty principle for stochastic currents.
\newblock {\em Phys. Rev. E}, 94:052104, 2016.

\bibitem{qian2013}
Hong Qian.
\newblock A decomposition of irreversible diffusion processes without detailed
  balance.
\newblock {\em Journal of Mathematical Physics}, 54(5):053302, 2013.

\bibitem{Renger2017a}
D.~R.~Michiel Renger.
\newblock Large deviations of specific empirical fluxes of independent {M}arkov
  chains, with implications for {M}acroscopic {F}luctuation {T}heory.
\newblock Weierstrass Institute, Preprint 2375, 2017.

\bibitem{Schnakenberg1976a}
J.~Schnakenberg.
\newblock Network theory of microscopic and macroscopic behavior of master
  equation systems.
\newblock {\em Rev. Modern Phys.}, 48(4):571--585, 1976.

\bibitem{Seifert2012a}
Udo Seifert.
\newblock Stochastic thermodynamics, fluctuation theorems and molecular
  machines.
\newblock {\em Rep. Prog. Phys.}, 75(12):126001, 2012.

\bibitem{Toner1995a}
John Toner and Yuhai Tu.
\newblock Long-range order in a two-dimensional dynamical $\mathrm{XY}$ model:
  How birds fly together.
\newblock {\em Phys. Rev. Lett.}, 75:4326--4329, 1995.

\bibitem{Touchette2009a}
Hugo Touchette.
\newblock The large deviation approach to statistical mechanics.
\newblock {\em Phys. Rep.}, 478(1-3):1--69, 2009.

\bibitem{vaik14}
S. Vaikuntanathan, T.~R.~Gingrich and P.~L.~Geissler.
\newblock Dynamic phase transitions in simple driven kinetic networks.
\newblock {\em Phys Rev. E}, 89:062108 (2014)

\bibitem{Wittkowski2014a}
Raphael Wittkowski, Adriano Tiribocchi, Joakim Stenhammar, Rosalind~J. Allen,
  Davide Marenduzzo, and Michael~E. Cates.
\newblock Scalar $\phi^4$ field theory for active-particle phase separation.
\newblock {\em Nat. Commun.}, 5:4351, 2014.

\end{thebibliography}
\def\cprime{$'$} \def\cprime{$'$} \def\cprime{$'$}
  \def\polhk#1{\setbox0=\hbox{#1}{\ooalign{\hidewidth
  \lower1.5ex\hbox{`}\hidewidth\crcr\unhbox0}}} \def\cprime{$'$}
  \def\cprime{$'$}

\end{document}